\documentclass[journal,onecolumn]{IEEEtran}

\usepackage[utf8]{inputenc}
\usepackage[T1]{fontenc}

\usepackage{amsmath,amssymb,amsfonts,amsthm}
\DeclareMathOperator{\spun}{Span}

\newtheorem{theorem}{Theorem}
\newtheorem{lemma}{Lemma}
\newtheorem{definition}{Definition}
\newtheorem{remark}{Remark}

\newtheorem{example}{Example}
\newtheorem{proposition}{Proposition}
\newtheorem{corollary}{Corollary}
\newtheorem{fact}{Fact}

\newtheorem{observation}{Observation}
\usepackage{graphicx}
\usepackage[table]{xcolor}
\usepackage{colortbl}
\usepackage{multirow}
\usepackage{float}
\usepackage{tikz}
\usetikzlibrary{
  matrix,positioning,arrows.meta,calc,backgrounds,
  decorations.pathreplacing,shapes,bending
}

\usepackage{nicefrac}
\usepackage{bbm}
\usepackage{mathrsfs}
\usepackage[mathscr]{eucal}
\usepackage{textcomp}
\usepackage{verbatim}
\usepackage{epstopdf}
\usepackage{stmaryrd}
\usepackage{mathabx}
\usepackage{framed}
\usepackage{mdframed}
\usepackage{adjustbox}

\usepackage[noadjust]{cite}

\usepackage[colorlinks,
            linkcolor=blue,
            anchorcolor=blue,
            citecolor=blue]{hyperref}
\definecolor{blockBlue}{RGB}{215,235,255}
\definecolor{blockGreen}{RGB}{215,255,215}
\definecolor{blockYellow}{RGB}{255,253,208}
\definecolor{blockRed}{RGB}{255,225,225}

\makeatletter
\@ifundefined{c@subfigure}{\newcounter{subfigure}[figure]}{}
\makeatother
\renewcommand{\thesubfigure}{\alph{subfigure}}

\newcommand{\Fq}{\mathbb{F}_{q}}
\newcommand{\C}{\mathcal{C}}
\newcommand{\E}{\mathbb{E}}
\newcommand{\F}{\mathbb{F}}
\newcommand{\K}{\mathbb{K}}
\newcommand{\R}{\mathcal{R}}
\newcommand{\I}{\mathcal{I}}

\newcommand{\Agg}{\mathrm{Agg}}
\newcommand{\Loc}{\mathsf{Loc}}
\newcommand{\ctr}{\mathsf{ctr}}
\newcommand{\Exp}{\mathsf{Exp}}

\begin{document}

\title{Reed-Solomon Codes with Optimal Repair Bandwidth: A Basis-Transformation Approach}

\author{%
Jing Qiu, Weijun Fang, Shu-Tao Xia, and Fang-Wei Fu%
\thanks{This work was supported in part by the National Key Research and Development Program of China under Grants 2021YFA1001000, 2022YFA1004900 and 2022YFA1005000; by the National Natural Science Foundation of China under Grants 62571298, 62571301, and 62371259;
by the Fundamental Research Funds for the Central Universities of China (Nankai University); by the Nankai Zhide Foundation; and by the China Postdoctoral Science Foundation under Grant 2025M773088. An earlier version of this article was presented at the 2026 IEEE International Symposium on Information Theory (ISIT)~\cite{isit2026}. (Corresponding author: Weijun Fang.)}
\thanks{Jing Qiu and Shu-Tao Xia are with the Tsinghua Shenzhen International Graduate School, Tsinghua University, Shenzhen, Guangdong 518055, China 
(e-mails: \{jing\_qiu, xiast\}@sz.tsinghua.edu.cn).}
\thanks{Weijun Fang is with the State Key Laboratory of Cryptography and Digital Economy Security, the Key Laboratory of Cryptologic Technology and Information Security, Ministry of Education, and the School of Cyber Science and Technology, Shandong University, Qingdao 266237, China (e-mail: fwj@sdu.edu.cn).}%
\thanks{Fang-Wei Fu is with the Chern Institute of Mathematics and LPMC, Nankai University, Tianjin 300071, China (e-mail: fwfu@nankai.edu.cn).}%
}
\maketitle

\begin{abstract}

Maximum distance separable (MDS) codes are widely used in distributed storage, but naively repairing a single failure in an $(n,k)$ MDS code requires downloading the full contents of $k$ surviving nodes. Minimum storage regenerating (MSR) codes, introduced by Dimakis et al., minimize repair bandwidth while preserving the MDS property by contacting $d>k$ helper nodes and downloading only a fraction of {the data stored at} each helper {node}. For scalar MDS codes, Guruswami and Wootters established a linear repair framework, and Tamo, Ye, and Barg subsequently gave the first explicit Reed-Solomon (RS) codes achieving the MSR point. Their construction yields RS-MSR codes with subpacketization $\ell=s\prod_{i=1}^n p_i$, where $s=d+1-k$ and the distinct primes $p_i$ satisfy $p_i\equiv 1\pmod{s}$. In this paper, we show that this congruence condition is not intrinsic to the RS repair problem. We develop a basis-transformation approach to the
construction of repair-enabling subspaces. The approach consists of three deterministic operations---Euclidean Square Partition, Transposition, and Column Aggregation---which construct the required repair-enabling subspaces directly from the standard monomial basis of the {extension} field. Consequently, we obtain RS-MSR codes with subpacketization $\ell=s\prod_{i=1}^n p_i$ for arbitrary distinct primes $p_i\ge s$. For fixed $s$, this improves the subpacketization of the Tamo--Ye--Barg construction by a factor asymptotic to $\varphi(s)^n\,e^{\mathrm{o}(n)}$, where $\varphi(\cdot)$ denotes Euler's totient function.

\end{abstract}

\section{Introduction}

\IEEEPARstart{R}{eed--Solomon} (RS) codes are widely used in distributed storage systems due to their optimal storage efficiency. A major practical challenge is that recovering a single failed node requires downloading a large amount of data from the remaining nodes. Dimakis et al.~\cite{Dimakis10} established the fundamental tradeoff between storage overhead and repair bandwidth, and Minimum Storage Regenerating (MSR) codes achieve the corresponding cut-set bound while maintaining the MDS property.

Assume the data are encoded by a code $\mathcal{C}$ over a finite field $F$, where each code symbol is either a single element of $F$ or an $\ell$-dimensional vector over $F$.
In the former case, $\mathcal{C}$ is called a \emph{scalar} code; in the latter case, $\mathcal{C}$ is called an \emph{array} code with subpacketization $\ell$.
Most MSR constructions have focused on array codes, and many interesting results have been obtained; see, e.g.,~\cite{Rashmi11PM, SuhRam11IA, Shah12Miser,Tamo13Zigzag,Cadambe13AsymIA,Papailio13Hadamard,WTB16ExplicitMSR, YeBarg17HighRate,Goparaju17AllParams,Sasidharan17CoupledLayer,Rawat16HighRateITA,Raviv15SmallFields,Vajha18Clay,ChenBarg20RackAware,Wang23TwoThird,Li2024linear}.
In contrast, constructing scalar MDS codes, particularly Reed-Solomon codes, that achieve the cut-set bound remained an open problem for years.
A landmark result by Tamo, Ye, and Barg~\cite{TYB17} provided the first explicit construction of $(n,k)$ Reed-Solomon codes that are MSR codes, utilizing the linear repair scheme for scalar codes proposed by Guruswami and Wootters~\cite{Guruswami16}.
As shown in~\cite{TYB17} and reviewed in Section~\ref{subsec:tyb17}, for a Reed-Solomon code over a finite field $\K$, optimal repair of a failed node~$i$ from $d$ helper nodes can be achieved by constructing an $F_i$-subspace (where $F_i$ is a subfield of $\K$) $S_i\subseteq\K$ satisfying
\begin{equation}\label{eq:key-subspace}
\dim_{F_i}(S_i)=p_i,\qquad \K=\sum_{u=0}^{s-1}{S_i}\alpha_i^u,
\end{equation}
where $s=d+1-k$. The construction in~\cite{TYB17}, which we refer to as the TYB construction, achieves~Eq.~\eqref{eq:key-subspace} under the arithmetic restriction $p_i\equiv1\pmod{s}$, and hence yields subpacketization $\ell=s\prod_{i=1}^n p_i$. \textbf{This raises a natural question: Is the congruence condition $p_i\equiv1\pmod{s}$ intrinsic to optimal repair, or merely an artifact of the basis choice in the TYB construction?}

In this paper, we adopt a \emph{basis-transformation} approach.  More precisely, we construct the required repair-enabling subspaces by applying an explicit change of basis, realized through three deterministic operations: Euclidean Square Partition, Transposition, and Column Aggregation. This new viewpoint allows us to remove the congruence restriction $p_i\equiv1\pmod{s}$ imposed in~\cite{TYB17}. Consequently, for arbitrary distinct primes $p_i\ge s$, we obtain RS-MSR codes with subpacketization
\[
\ell=s\prod_{i=1}^n p_i .
\] For fixed~$s$, this improves the subpacketization of~\cite{TYB17} by a factor asymptotic to $\varphi(s)^n\,e^{\mathrm{o}(n)}$ (see Table~\ref{tab:asymptotic-subpacketization}).

\begin{table}[t]
\centering
\caption{Asymptotic subpacketization for fixed $s$}
\label{tab:asymptotic-subpacketization}
\renewcommand{\arraystretch}{1.12}
\begin{tabular}{|c|c|}
\hline
Construction & Subpacketization \\
\hline
RS-MSR code in~\cite{TYB17} & $(\varphi(s)n)^{n+\mathrm{o}(n)}$ \\
\hline
RS-MSR code in this paper & $n^{n+\mathrm{o}(n)}$ \\
\hline
\end{tabular}
\end{table}

The remainder of this paper is organized as follows. Section~\ref{sec:pre} provides the necessary preliminaries on RS codes and linear repair schemes. Section~\ref{sec:main} presents our main results. Section~\ref{sec:conclusion} concludes this paper.

\section{Preliminaries}\label{sec:pre}

Throughout this paper, we use the following notation unless stated otherwise.
\begin{itemize}
    \item For integers \(a,b\), define
\[
[a,b]:=\{a,a+1,\dots,b\},
\]
with the convention that \([a,b]=\emptyset\) whenever \(a>b\).

    \item For a positive integer $n$, define $[n] := \{1,\dots,n\}$.
    \item For an array $A$ of size $a\times b$, we index its entries by
    $(i,j)\in \mathcal{I}(a,b)$, where
    \[
        \mathcal{I}(a,b) := [0,a-1]\times[0,b-1],
    \]
    and denote by $A(i,j)$ the entry in row $i$ and column $j$.
\end{itemize}

\subsection{RS Codes and Repair Bandwidth}

We first recall how a Reed-Solomon code over an extension field can be viewed as an array code over its base field.

Let $\K$ be a finite field. Let $\Omega=\{\omega_1,\omega_2,\ldots,\omega_n\}\subseteq \K$ be a set of $n$ pairwise distinct elements, and let $v=(v_1,\ldots,v_n)\in (\K^*)^n$, where $\K^*:=\K\setminus\{0\}$. 
The generalized Reed-Solomon code $\operatorname{GRS}_{\K}(n,k,\Omega,v)\subseteq \K^n$ with evaluation set $\Omega$ and column multipliers $v$ is defined as the set
\[
\left\{\left(v_1 f(\omega_1),\ldots,v_n f(\omega_n)\right)\in \K^n:\,f\in \K[x],\ \deg f\le k-1\right\}.
\]
If $v=(1,\ldots,1)$, then the GRS code is called a Reed-Solomon code and is denoted by $\operatorname{RS}_{\K}(n,k,\Omega)$.
It is well known (see \cite{MacWilliams77}) that
\[
\left(\operatorname{RS}_{\K}(n,k,\Omega)\right)^{\perp}
=\operatorname{GRS}_{\K}(n,n-k,\Omega,u),
\]
where $u_i=\prod_{j\neq i}(\omega_i-\omega_j)^{-1}$ for $i=1,\ldots,n$.

Let $\C$ be an $(n,k)$ MDS array code over a finite field $\F$ with subpacketization $\ell$, i.e., each node stores a vector $c_i\in \F^{\ell}$. If the $i$-th node $c_i$ is erased, the repair bandwidth $N(\C, i,\mathcal{H})$ of $c_i$ is defined as the total amount of information (measured in symbols over $\F$) downloaded from a helper set $\mathcal{H}\subseteq [n]\setminus\{i\}$ with $|\mathcal{H}|=d$ to reconstruct $c_i$. 
The cut-set bound of \cite{Dimakis10} states that the repair bandwidth $N(\C,i,\mathcal{H})$ satisfies
\begin{equation}\label{eq:cutset}
N(\C, i, \mathcal{H}) \ge \frac{d\ell}{d-k+1}.
\end{equation}
We call $\C$ an MSR code if it achieves the cut-set bound in Eq.~\eqref{eq:cutset} with equality for all $i \in [n]$ and every helper set \(\mathcal{H}\subseteq[n]\setminus\{i\}\) with \(|\mathcal{H}|=d\).

If $\K/\F$ is an extension of degree $\ell$, then every RS code over $\K$ can be regarded as an MDS array code over $\F$ with subpacketization $\ell$. An RS code over $\K$ is called an RS-MSR code if it is an MSR code under this interpretation.

\subsection{Reed-Solomon codes as MSR codes}\label{subsec:tyb17}

The next theorem converts the linear repair of a scalar MDS code into a dual-code criterion. It is the starting point for both the TYB construction and our new construction.

\begin{theorem}[\cite{Guruswami16}] Let $\C \subset \K^n$ be a scalar linear MDS
code of length $n$. Let $\F$ be a subfield of $\K$ such that $[\K :
\F ]=\ell$. For a given $i \in [n]$, the following statements
are equivalent.

\begin{enumerate}
  \item[(1)] There is a linear repair scheme of the node $c_i$ over $\F$
such that the repair bandwidth $N(\C, i,  {[n]\setminus\{i\}}) \leq b$.
  \item[(2)] There is a subset of codewords $P_i\subseteq \C^{\perp}$ with size $|P_i|=\ell$ such that
$${\rm dim}_{\F}(\operatorname{Span}_{\mathbb F}\{x_i:x\in P_i\})=\ell,$$
and
$$b\geq \sum_{j\in {[n]\setminus\{i\}}}{\rm dim}_{\F}(\operatorname{Span}_{\mathbb F}\{x_j:x\in P_i\}).$$
\end{enumerate}

\end{theorem}
We next recall the TYB construction, which formulates a sufficient subspace condition for optimal repair. Our basis-transformation approach, introduced in Section~\ref{sec:main}, will provide a new constructive procedure for satisfying this condition.


Let $\Fq$ be a finite field.
Denote $s:=d-k+1$ and let $p_1,\dots,p_n$ be $n$ distinct primes such that
\begin{equation}\label{eq:pms}
p_i\equiv 1 {\pmod{s}} \;\;\text{~for all~} i=1,2,\dots,n.
\end{equation}
For $i=1,\dots,n$, let $\alpha_i$ be an element of an extension field of $\Fq$ such that $[\Fq(\alpha_i):\Fq]=p_i$, and let
\begin{equation*}
\E:=\Fq(\alpha_1,\dots,\alpha_n).
\end{equation*}
Since the primes
$p_1,\ldots,p_n$ are distinct, $[\E:\Fq]=\prod_{i=1}^np_i$.
For $i=1,\dots,n$, define
\[
F_i=\Fq(\{\alpha_j:j\neq i\}).
\]
Then $\E= F_i(\alpha_i)$ and $[\E : F_i]=p_i$.
Let $\K$ be the degree-$s$ extension of the field $\E$, i.e., there exists an element $\beta\in \K$ of degree $s$ over $\E$ such that $\K=\E(\beta)$.
We also have $[\K: F_i]=s p_i$ for all $i$.

The following theorem gives the exact TYB construction under the congruence condition~Eq.~\eqref{eq:pms}.
\begin{theorem}[\cite{TYB17}]\label{thm1}
Let $k,n,d$ be any positive integers such that $k< d \leq n-1.$ Let $\Omega=\{\alpha_1,\dots,\alpha_n\}$, where each $\alpha_i$, for $i=1,\dots,n$, is an element of degree $p_i$
over $\Fq$ and $p_i$ is a prime satisfying Eq.~\eqref{eq:pms}.
 The code $\C:=\text{\rm RS}_{\K}(n,k,\Omega)$ achieves the cut-set bound for the repair of any single node from any $d$ helper nodes.
In other words, $\C$ is an $(n,k)$ MSR code with repair degree $d$ and subpacketization $s\prod_{i=1}^{n}p_i$.
 \end{theorem}

The following trace-expansion identity is fundamental to trace-based repair schemes for Reed-Solomon codes.
\begin{lemma}[\cite{Lidl94}]\label{lem:expression}
Let $\K$ be an extension field of a finite field $\F$ with $[\K:\F]=\ell$.
Assume $\{\zeta_1,\dots,\zeta_{\ell}\}$ and $\{\epsilon_1,\dots,\epsilon_{\ell}\}$ are dual bases of $\K$ over $\F$, i.e., ${\rm tr}(\zeta_i\epsilon_j)=\delta_{ij}$, where ${\rm tr}(\cdot)$ denotes the trace function from $\K$ to $\F$ and $\delta_{ij}$ is the Kronecker delta.
Then for any $\varepsilon\in \K$,
      $$\varepsilon=\sum_{i=1}^{\ell}{\rm tr}(\zeta_i\varepsilon)\epsilon_{i}.$$

\end{lemma}

Under the congruence condition~Eq.~\eqref{eq:pms}, Tamo, Ye, and Barg~\cite{TYB17} provide the following optimal repair-enabling subspace for their construction.
 \begin{lemma}[{\cite{TYB17}}]\label{lem1}
Let $\alpha_i,\beta,F_i$ be as above.
Given $i\in [n]$, define the following vector spaces over $F_i$:
\begin{align*}
S_i^{(1)} &=\spun_{F_i} \big(\beta^u \alpha_i^{u+rs}: u\in [0,s-1],\ r\in [0,\tfrac{p_i-1}{s}-1] \big){,} \\
S_i^{(2)} &=\spun_{F_i} \Big(\sum_{u=0}^{s - 1}\beta^u \alpha_i^{p_i-1} \Big){,} \\
S_i &=S_i^{(1)} + S_i^{(2)}{,}
\end{align*}
where $(p_i-1)/s\in\mathbb{Z}$ by~Eq.~\eqref{eq:pms}.

Then
\begin{equation}\label{eq:S}
    \dim_{F_i} S_i = p_i \text{ and }
S_i + S_i\alpha_i+\dots + S_i\alpha_i^{s-1}=\K.
\end{equation}

\end{lemma}
We next explain how the subspace condition~Eq.~\eqref{eq:S}, together with the trace-expansion identity, enables the repair process.

Let $\C = {\rm RS}_{\K}(n,k,\Omega)$ and let $c=(c_1,c_2,\dots,c_n)$ be a codeword of $\C$. Assume that $c_i$ is erased. The TYB repair argument proceeds as follows.

\begin{enumerate}

  \item By Lemma~\ref{lem1}, there exists an $F_i$-subspace $S_i\subseteq\K$ such that
\[
\dim_{F_i} S_i= p_i \qquad\text{and}\qquad \K=S_i +S_i\alpha_i +\dots+ S_i\alpha_i^{s-1}.
\]
Let $e_1, \dots, e_{p_i}$ be an arbitrary basis of the subspace $S_i$ over the field $F_i$.

  \item Let $\mathcal H=\{i_1,\dots,i_d\}\subset [n]\setminus \{i\}$ be any helper set of size $d =s+ k-1$.

  \item Let
\[
\Lambda(x)=\prod_{j\in [n]\setminus(\mathcal H\cup \{i\})}(x-\alpha_j).
\]
Then for each $u\in[0,s-1]$, the vector
\[
(u_1\alpha_1^{u}\Lambda(\alpha_1),\dots , u_n\alpha_n^{u}\Lambda(\alpha_n))
\]
belongs to $\C^{\perp}$ because $\deg(x^u\Lambda(x))\le (s-1)+(n-d-1)=n-k-1$. Hence
\[
c_i u_i\alpha_i^u \Lambda(\alpha_i)+\sum_{j=1}^{d}c_{i_j}u_{i_j}\alpha_{i_j}^{u}\Lambda(\alpha_{i_j})=0.
\]

  \item Let ${\rm tr}_i := {\rm tr}_{\K/F_i}$. Multiplying the above identity by $e_m$ and applying ${\rm tr}_i$ yields, for each $u\in[0,s-1]$ and $m\in[1,p_i]$,
  \[
{\rm tr}_i(c_i e_m u_i \alpha_i^{u} \Lambda(\alpha_i))=-\sum_{j=1}^{d}{\rm tr}_i(c_{i_j} e_m u_{i_j} \alpha_{i_j}^{u} \Lambda(\alpha_{i_j}))=-\sum_{j=1}^{d}\alpha_{i_j}^{u} \Lambda(\alpha_{i_j}){\rm tr}_i(c_{i_j} e_m u_{i_j}),
\]
because $\alpha_{i_j},\Lambda(\alpha_{i_j})\in F_i$ whenever $i_j\neq i$.

  \item Since $\K=\sum_{u=0}^{s-1}{S_i}\alpha_i^u$ and $[\K:F_i]=sp_i$, the $sp_i$ elements
\[
\left\{e_m\alpha_i^u:\ 0\le u\le s-1,\ 1\le m\le p_i\right\}
\]
form an $F_i$-basis of $\K$; multiplying by the nonzero scalar $u_i \Lambda(\alpha_i)$ preserves this basis. Therefore Lemma~\ref{lem:expression} shows that the traces
\[
\left\{{\rm tr}_i(c_i e_m u_i\alpha_i^{u}\Lambda(\alpha_i))\right\}_{0\le u\le s-1,\ 1\le m\le p_i}
\]
determine $c_i$.
  \item From each helper node $i_j$, where $1\leq j\leq d$, the repair procedure downloads precisely
\[
\left\{\operatorname{tr}_i(c_{i_j} e_m u_{i_j})\right\}_{m=1}^{p_i},
\]
which are $p_i$ symbols over $F_i$. Since $[F_i:\Fq]=\prod_{j\ne i}p_j$ and $\ell=[\K:\Fq]=s\prod_{j=1}^{n}p_j$, the total repair bandwidth, measured in $\Fq$-symbols, is
\[
d\cdot p_i\cdot [F_i:\Fq]
=d\prod_{j=1}^{n}p_j
=\frac{d\ell}{s}
=\frac{d\ell}{d-k+1},
\]
which meets the cut-set bound~Eq.~\eqref{eq:cutset}.
\end{enumerate}

\begin{remark}\label{rem:key-obstacle}
The argument above shows that the entire repair problem reduces to constructing, for each failed node~$i$, an $F_i$-subspace $S_i\subseteq\K$ satisfying~Eq.~\eqref{eq:S}. In~\cite{TYB17}, Lemma~\ref{lem1} provides such subspaces only under the congruence condition $p_i\equiv 1\pmod{s}$. In Section~\ref{sec:main}, we remove this restriction by constructing the required subspaces via an explicit basis{-}transformation technique.
\end{remark}

\section{Main results}\label{sec:main}

In this section, we develop a basis-transformation technique to construct
an optimal repair-enabling subspace satisfying the analogue of Eq.~\eqref{eq:S}, namely Eq.~\eqref{eq:smain} under the sole condition $p\ge s$.

We fix a tower of finite fields $\F\subset\E\subset\K$ with
$[\E:\F]=p$ and $[\K:\E]=s$, where $p\ge s$. Let $\alpha$ generate $\E$
over~$\F$ and $\beta$ generate $\K$ over~$\E$. In the RS-MSR setting
of Section~\ref{subsec:tyb17}, this tower corresponds to
$F_i\subset\E\subset\K$ with $\alpha=\alpha_i$ for each node~$i$. As
discussed in Section~\ref{subsec:tyb17} (see
Remark~\ref{rem:key-obstacle}), the key step in optimal repair of
Reed-Solomon codes is to construct an $\F$-subspace~$S$ of~$\K$ such
that~Eq.~\eqref{eq:smain} holds.
To this end, we introduce a technique composed of three
operations---Euclidean Square Partition, Transposition, and
Column Aggregation---that transform the standard monomial basis array of~$\K$
into a new basis array with the specified property.
The $\F$-span of the first row of the transformed basis array is then
exactly the desired subspace~$S$.
Since all three operations are governed by the Euclidean algorithm
for~$(p,s)$, we fix the following notation throughout.

\begin{itemize}
    \item {Set $b_0:=s$ and let $m\in\mathbb Z_{\ge0}$.} Write the Euclidean algorithm for~$(p,s)$ as
    \begin{equation}\label{eq:euclid}
    \begin{aligned}
    p &= a_0 s + b_1,\\
    s &= a_1 b_1 + b_2,\\
    b_1 &= a_2 b_2 + b_3,\\
    &\ \vdots\\
    b_{m-2} &= a_{m-1} b_{m-1} + b_m,\\
    b_{m-1} &= a_m b_m,
    \end{aligned}
    \end{equation}
    where $b_m=\gcd(p,s)$. Set $\nu:=\lfloor m/2\rfloor$ and $\nu':=\lceil m/2\rceil$.

    \item Define $C_0:=0$, and
\[
      C_r:=\sum_{\mu=1}^{r} a_{2\mu}\, b_{2\mu},\quad 1\le r\le \nu.
\]

    \item Define $Y_0:=0$, and
\[
    Y_h:=\sum_{\gamma=1}^{h} a_{2\gamma-1}\, b_{2\gamma-1},\quad 1\le h\le \nu'.
\]

    \item For $j\in[0,b_1-1]$, define
    \begin{equation}\label{eq:j-iter}
    j^{(0)}:=j, \qquad j^{(2r)}:=j^{(2r-2)} \pmod{b_{2r}}, \quad 1\leq r\leq \nu.
    \end{equation}
\end{itemize}
{Our} main {result is stated} as follows.

\begin{theorem}\label{thm:main}
Let $\F\subset\E\subset\K$ be finite fields with $[\E:\F]=p\ge s=[\K:\E]$.
Fix $\alpha\in\E$ and $\beta\in\K$ such that $\E=\F(\alpha)$ and $\K=\E(\beta)$. Then, there exists an $\F$-subspace~$S$ of~$\K$ such
that
\begin{equation}\label{eq:smain}
  \dim_{\F} S = p \quad\text{and}\quad
  S + S\alpha + \dots + S\alpha^{s-1} = \K.
\end{equation}
Furthermore, the explicit structure of $S$ is given as follows: 
\[S   =S^{(1)} + S^{(2)},\]
where
\begin{align*}
S^{(1)} &=\spun_{\F} \bigl(\alpha^{rs}(\alpha\beta)^{u}:\, r\in [0,a_0-1],\, u\in [0,s-1] \bigr), \\
S^{(2)} &=\spun_{\F} \bigl(\R_j:\, j\in [0,b_1-1] \bigr), 
\end{align*}
and for each $j\in[0,b_1-1]$, the element $\R_j\in\K$ is defined as follows.
\begin{enumerate}
\item[(1)] If $C_{r-1}\le j < C_r$ for some $1\le r\le \nu$, then
\[
\R_j={}\alpha^{a_0 s+j}\!\Biggl( \sum_{h=0}^{r-1}\sum_{w=0}^{a_{2h+1}-1} \beta^{Y_h +w\,b_{2h+1}+j^{(2h)}}+ \beta^{Y_r+j^{(2r)}} \Biggr).
\]
\item[(2)] If  $C_{\nu}\le j\le b_1-1$ and $m$ is odd, then
\[
\R_j
=
\alpha^{a_0s+j}
\sum_{h=0}^{\nu'-1}\sum_{w=0}^{a_{2h+1}-1}
\beta^{Y_h
  +w\,b_{2h+1}+j^{(2h)}}.
\]
\end{enumerate}
\end{theorem}

\begin{corollary}\label{cor:m1}
When $m=1$, from Theorem~\ref{thm:main} we have that for $j\in [0,b_1-1]$,
$\R_j
=\alpha^{a_0s}(\alpha\beta)^j\sum_{t=0}^{a_1-1}\beta^{t b_1}$.
\end{corollary}
\begin{proof}
When $m=1$, we have $C_0=0\le j\le b_1-1$, so every column falls into case~(2) of Theorem~\ref{thm:main} with a single index $h=0$ in the displayed sum. This gives the exponents $w\,b_1+j$ for $w=0,\dots,a_1-1$. Factoring out $\alpha^{a_0s+j}\beta^j=\alpha^{a_0s}(\alpha\beta)^j$ yields the formula.
\end{proof}

\begin{remark}
The construction in \cite{TYB17} (see Lemma~\ref{lem1}) is a special case of Corollary \ref{cor:m1} with $b_1=1$.
\end{remark}

\begin{corollary}\label{cor:m2}
When $m=2$, Theorem~\ref{thm:main} gives
\[
\R_j =\alpha^{a_0s}(\alpha\beta)^j\Bigg(\sum_{w=0}^{a_1-1}\beta^{w b_1} +\beta^{a_1b_1-\left\lfloor j/b_2\right\rfloor b_2}\Bigg),\qquad j\in [0,b_1-1].
\]
\end{corollary}
\begin{proof}
When $m=2$, we have $C_1=a_2b_2=b_1$. Since $C_0=0\le j<C_1=b_1$, every column falls into case~(1) of Theorem~\ref{thm:main} with $r=1$. The double sum contributes the exponents $w\,b_1+j$ for $w=0,\dots,a_1-1$, and the final term contributes $a_1b_1+j^{(2)}$ where $j^{(2)}=j\pmod {b_2}$. Factoring out $\alpha^{a_0s}(\alpha\beta)^j$ and using $j-j^{(2)}=\lfloor j/b_2\rfloor b_2$ yields the formula.
\end{proof}
\begin{remark}
The example in Fig.~\ref{fig:boxed-tikz} will also illustrate the formula in Corollary~\ref{cor:m2}.
\end{remark}

\subsection{A basis-transformation perspective}\label{subsec:basistrans}

In this subsection, we formulate the basis-transformation problem. Given the standard monomial basis of~$\K$ arranged as a $p\times s$ array, the goal is to transform it into an $s\times p$ array whose rows are successive $\alpha$-multiples of the first row. The $\F$-span of the first row then serves as the optimal repair-enabling subspace~$S$.

We first introduce the monomial basis of $\K$ over $\F$. {Since} $\K=\E(\beta)$ and $\E=\F(\alpha)$, $\{\beta^u:0\le u\le s-1\}$ is an $\E$-basis of~$\K$ and $\{\alpha^i:0\le i\le p-1\}$ is an $\F$-basis of~$\E$. Since $\alpha\in\E^*$, we have $\E(\alpha\beta)=\E(\beta)=\K$, so $\{(\alpha\beta)^u:0\le u\le s-1\}$ is also an $\E$-basis of~$\K$. Therefore, the set
\begin{equation*}
 \bigl\{ \alpha^{i}(\alpha\beta)^{j} : i=0,\dots,p-1;~ j=0,\dots,s-1 \bigr\}
\end{equation*}
is an $\F$-basis of $\K$. We arrange it as a $p\times s$ array~$B$ (see Fig.~\ref{fig:basis0}), specifically, for each $(i,j)\in\I(p,s)$,
\begin{equation}\label{eq:B-def}
B(i,j)=\alpha^{i}(\alpha\beta)^{j}.
\end{equation}
Accordingly, the problem of finding $S$ satisfying~Eq.~\eqref{eq:smain} reduces to transforming the monomial basis array $B$ into a new $\F$-basis that can be arranged as an $s \times p$ array, whose rows are successive $\alpha$-multiples of the first row. The desired subspace $S$ is then given by the $\F$-span of the first row (see Fig.~\ref{fig:basis1}). The target arrangement in Fig.~\ref{fig:basis1} is meaningful only when $p\ge s$, because the number of successive $\alpha$-layers cannot exceed the dimension $p=[\E:\F]$.
\begin{figure*}[!t]
    \centering
    \small
    \renewcommand{\arraystretch}{1.4}
    \setlength{\arraycolsep}{4pt}

    \begin{minipage}[c]{0.46\linewidth}
        \centering
        \begin{adjustbox}{max width=\linewidth}
        $
        \begin{array}{|r|r|c|r|}
            \multicolumn{1}{c}{\mathbb{E}} &
            \multicolumn{1}{c}{\beta\mathbb{E}} &
            \multicolumn{1}{c}{\cdots} &
            \multicolumn{1}{c}{\beta^{s-1}\mathbb{E}} \\
            \multicolumn{1}{c}{\downarrow} &
            \multicolumn{1}{c}{\downarrow} &
            \multicolumn{1}{c}{} &
            \multicolumn{1}{c}{\downarrow} \\
            \hline
            1 & \alpha\beta & \cdots & (\alpha\beta)^{s-1} \\ \hline
            \alpha & \alpha(\alpha\beta) & \cdots & \alpha(\alpha\beta)^{s-1} \\ \hline
            \alpha^2 & \alpha^2(\alpha\beta) & \cdots & \alpha^2(\alpha\beta)^{s-1} \\ \hline
            \alpha^3 & \alpha^3(\alpha\beta) & \cdots & \alpha^3(\alpha\beta)^{s-1} \\ \hline
            \vdots & \vdots & \ddots & \vdots \\ \hline
            \alpha^{p-1} & \alpha^{p-1}(\alpha\beta) & \cdots & \alpha^{p-1}(\alpha\beta)^{s-1} \\ \hline
        \end{array}
        $
        \end{adjustbox}
        \caption{Original $p \times s$ array}
        \label{fig:basis0}
    \end{minipage}
    \hfill
    \begin{minipage}[c]{0.06\linewidth}
        \centering
        {\Large $\xrightarrow{?}$}
    \end{minipage}
    \hfill
    \begin{minipage}[c]{0.46\linewidth}
        \centering
        \begin{adjustbox}{max width=\linewidth}
        $
        \begin{array}{|r|r|r|c|r| l}
            \cline{1-5}
            x_0 & x_1 & x_2 & \cdots & x_{p-1} & \leftarrow S \\ \cline{1-5}
            \alpha x_0 & \alpha x_1 & \alpha x_2 & \cdots & \alpha x_{p-1} & \leftarrow {S}\alpha \\ \cline{1-5}
            \vdots & \vdots & \vdots & \ddots & \vdots & \vdots \\ \cline{1-5}
            \alpha^{s-1} x_0 & \alpha^{s-1} x_1 & \alpha^{s-1} x_2 & \cdots & \alpha^{s-1} x_{p-1} & \leftarrow {S}\alpha^{s-1} \\ \cline{1-5}
        \end{array}
        $
        \end{adjustbox}
        \caption{Target $s \times p$ array}
        \label{fig:basis1}
    \end{minipage}
\end{figure*}

\subsection{An illustrative example of the basis{-}transformation technique}

We {will} present a basis-transformation technique that achieves the conversion described above. Before introducing the formal definitions, we illustrate the technique with the example $(p,s)=(7,5)$, whose entire transformation procedure is depicted in Fig.~\ref{fig:boxed-tikz}. Applying the Euclidean algorithm yields $7=1\times 5+2$, $5=2\times 2+1$, and $2=2\times 1$. The transformation proceeds as follows:

\textbf{Step 1.} Partition $B$ into square blocks whose side lengths are determined by the Euclidean algorithm for $(7,5)$, yielding one $5 \times 5$ block, two $2 \times 2$ blocks, and two $1 \times 1$ blocks.

\textbf{Step~2.1.} Transpose the entire array $B$ to obtain a $5\times 7$ array~$B^{\top}$.

\textbf{Step~2.2.} Transpose each square block in $B^{\top}$ individually. The resulting $5\times 7$ array is denoted by~$R$.

\textbf{Step 3.} For each column of~$R$, we keep the $\alpha$-part of each entry unchanged and replace the $\beta$-part by the sum of all distinct powers of~$\beta$ appearing in that column. Specifically, the columns with indices $0$ through $4$ remain unchanged. In the column with index $5$, the $\beta$-part of each entry is replaced by $1+\beta^2+\beta^4$, and in the column with index $6$, the $\beta$-part of each entry is replaced by $\beta+\beta^3+\beta^4$. The resulting array is denoted by~$\overline R$.

Then the rows of $\overline R$ are successive $\alpha$-multiples of the first row. Let $S$ be the $\F$-subspace spanned by the first row. To verify~Eq.~\eqref{eq:smain} in this example, it remains to show that
\[
S+S\alpha+S\alpha^2+S\alpha^3+S\alpha^4=\K.
\]
Equivalently, it suffices to show that every entry of $R$ is an $\F$-linear combination of entries of $\overline R$.

\begin{figure*}[!t]
\centering
\begin{minipage}{0.96\textwidth}

\centering

\def\vsepA{1.1cm}
\def\vsepB{1.1cm}

\begin{adjustbox}{max width=\linewidth}
\begin{tikzpicture}[
    node distance=1.0cm,
    mat/.style={
        matrix of nodes,
        nodes in empty cells,
        nodes={
            anchor=center,
            text height=1.6ex,
            text depth=0.25ex,
            minimum height=1.4em,
            minimum width=2.8em,
            inner sep=1.5pt,
            font=\scriptsize,
            execute at begin node=$,
            execute at end node=$,
            draw=black!80,
            thin
        },
        column sep=-\pgflinewidth,
        row sep=-\pgflinewidth,
        draw=black,
        thick,
        inner sep=0pt
    },
    arrowLine/.style={
        ->,
        >={Stealth[length=3mm, width=2mm]},
        line width=1.5pt,
        draw=gray!50
    },
    titleNode/.style={
        font=\bfseries\small,
        text=black!80,
        align=center,
        inner sep=4pt
    },
    scale=0.82,
    transform shape
]

\colorlet{tinyBlockA}{blockRed}
\colorlet{tinyBlockB}{blockRed!55!blockYellow}

\matrix (m2) [mat] {
|[fill=blockBlue]| 1 & |[fill=blockBlue]| \alpha\beta & |[fill=blockBlue]| \alpha^{2}\beta^{2} & |[fill=blockBlue]| \alpha^{3}\beta^{3} & |[fill=blockBlue]| \alpha^{4}\beta^{4} & |[fill=blockGreen]| \alpha^{5} & |[fill=blockGreen]| \alpha^{6}\beta \\
|[fill=blockBlue]| \alpha & |[fill=blockBlue]| \alpha^{2}\beta & |[fill=blockBlue]| \alpha^{3}\beta^{2} & |[fill=blockBlue]| \alpha^{4}\beta^{3} & |[fill=blockBlue]| \alpha^{5}\beta^{4} & |[fill=blockGreen]| \alpha^{6} & |[fill=blockGreen]| \alpha^{7}\beta \\
|[fill=blockBlue]| \alpha^{2} & |[fill=blockBlue]| \alpha^{3}\beta & |[fill=blockBlue]| \alpha^{4}\beta^{2} & |[fill=blockBlue]| \alpha^{5}\beta^{3} & |[fill=blockBlue]| \alpha^{6}\beta^{4} & |[fill=blockYellow]| \alpha^{7}\beta^{2} & |[fill=blockYellow]| \alpha^{8}\beta^{3} \\
|[fill=blockBlue]| \alpha^{3} & |[fill=blockBlue]| \alpha^{4}\beta & |[fill=blockBlue]| \alpha^{5}\beta^{2} & |[fill=blockBlue]| \alpha^{6}\beta^{3} & |[fill=blockBlue]| \alpha^{7}\beta^{4} & |[fill=blockYellow]| \alpha^{8}\beta^{2} & |[fill=blockYellow]| \alpha^{9}\beta^{3} \\
|[fill=blockBlue]| \alpha^{4} & |[fill=blockBlue]| \alpha^{5}\beta & |[fill=blockBlue]| \alpha^{6}\beta^{2} & |[fill=blockBlue]| \alpha^{7}\beta^{3} & |[fill=blockBlue]| \alpha^{8}\beta^{4} & |[fill=tinyBlockA]| \alpha^{9}\beta^{4} & |[fill=tinyBlockB]| \alpha^{10}\beta^{4} \\
};
\node[titleNode, above=1pt] at (m2.north) {Step 2.2: Block-wise Transposition. $R$};

\matrix (m1) [mat, above=\vsepA of m2] {
|[fill=blockBlue]| 1 & |[fill=blockBlue]| \alpha & |[fill=blockBlue]| \alpha^{2} & |[fill=blockBlue]| \alpha^{3} & |[fill=blockBlue]| \alpha^{4} & |[fill=blockGreen]| \alpha^{5} & |[fill=blockGreen]| \alpha^{6} \\
|[fill=blockBlue]| \alpha\beta & |[fill=blockBlue]| \alpha^{2}\beta & |[fill=blockBlue]| \alpha^{3}\beta & |[fill=blockBlue]| \alpha^{4}\beta & |[fill=blockBlue]| \alpha^{5}\beta & |[fill=blockGreen]| \alpha^{6}\beta & |[fill=blockGreen]| \alpha^{7}\beta \\
|[fill=blockBlue]| \alpha^{2}\beta^{2} & |[fill=blockBlue]| \alpha^{3}\beta^{2} & |[fill=blockBlue]| \alpha^{4}\beta^{2} & |[fill=blockBlue]| \alpha^{5}\beta^{2} & |[fill=blockBlue]| \alpha^{6}\beta^{2} & |[fill=blockYellow]| \alpha^{7}\beta^{2} & |[fill=blockYellow]| \alpha^{8}\beta^{2} \\
|[fill=blockBlue]| \alpha^{3}\beta^{3} & |[fill=blockBlue]| \alpha^{4}\beta^{3} & |[fill=blockBlue]| \alpha^{5}\beta^{3} & |[fill=blockBlue]| \alpha^{6}\beta^{3} & |[fill=blockBlue]| \alpha^{7}\beta^{3} & |[fill=blockYellow]| \alpha^{8}\beta^{3} & |[fill=blockYellow]| \alpha^{9}\beta^{3} \\
|[fill=blockBlue]| \alpha^{4}\beta^{4} & |[fill=blockBlue]| \alpha^{5}\beta^{4} & |[fill=blockBlue]| \alpha^{6}\beta^{4} & |[fill=blockBlue]| \alpha^{7}\beta^{4} & |[fill=blockBlue]| \alpha^{8}\beta^{4} & |[fill=tinyBlockA]| \alpha^{9}\beta^{4} & |[fill=tinyBlockB]| \alpha^{10}\beta^{4} \\
};
\node[titleNode, above=1pt] at (m1.north) {Step 2.1: Overall Transposition. $B^{\top}$};

\matrix (m3) [
    mat,
    below=\vsepB of m2,
    column 6/.style={nodes={minimum width=6.8em}},
    column 7/.style={nodes={minimum width=6.8em}}
] {
|[fill=blockBlue]| 1 & |[fill=blockBlue]| \alpha\beta & |[fill=blockBlue]| \alpha^{2}\beta^{2} & |[fill=blockBlue]| \alpha^{3}\beta^{3} & |[fill=blockBlue]| \alpha^{4}\beta^{4} & |[fill=blockGreen]| \alpha^{5}(1+\beta^{2}+\beta^{4}) & |[fill=blockGreen]| \alpha^{6}(\beta+\beta^{3}+\beta^{4}) \\
|[fill=blockBlue]| \alpha & |[fill=blockBlue]| \alpha^{2}\beta & |[fill=blockBlue]| \alpha^{3}\beta^{2} & |[fill=blockBlue]| \alpha^{4}\beta^{3} & |[fill=blockBlue]| \alpha^{5}\beta^{4} & |[fill=blockGreen]| \alpha^{6}(1+\beta^{2}+\beta^{4}) & |[fill=blockGreen]| \alpha^{7}(\beta+\beta^{3}+\beta^{4}) \\
|[fill=blockBlue]| \alpha^{2} & |[fill=blockBlue]| \alpha^{3}\beta & |[fill=blockBlue]| \alpha^{4}\beta^{2} & |[fill=blockBlue]| \alpha^{5}\beta^{3} & |[fill=blockBlue]| \alpha^{6}\beta^{4} & |[fill=blockYellow]| \alpha^{7}(1+\beta^{2}+\beta^{4}) & |[fill=blockYellow]| \alpha^{8}(\beta+\beta^{3}+\beta^{4}) \\
|[fill=blockBlue]| \alpha^{3} & |[fill=blockBlue]| \alpha^{4}\beta & |[fill=blockBlue]| \alpha^{5}\beta^{2} & |[fill=blockBlue]| \alpha^{6}\beta^{3} & |[fill=blockBlue]| \alpha^{7}\beta^{4} & |[fill=blockYellow]| \alpha^{8}(1+\beta^{2}+\beta^{4}) & |[fill=blockYellow]| \alpha^{9}(\beta+\beta^{3}+\beta^{4}) \\
|[fill=blockBlue]| \alpha^{4} & |[fill=blockBlue]| \alpha^{5}\beta & |[fill=blockBlue]| \alpha^{6}\beta^{2} & |[fill=blockBlue]| \alpha^{7}\beta^{3} & |[fill=blockBlue]| \alpha^{8}\beta^{4} & |[fill=tinyBlockA]| \alpha^{9}(1+\beta^{2}+\beta^{4}) & |[fill=tinyBlockB]| \alpha^{10}(\beta+\beta^{3}+\beta^{4}) \\
};
\node[titleNode, above=1pt] at (m3.north) {Step 3: Column Aggregation. $\overline{R}$};

\matrix (m0) [mat, left=3cm of m2] {
|[fill=blockBlue]| 1 & |[fill=blockBlue]| \alpha\beta & |[fill=blockBlue]| \alpha^{2}\beta^{2} & |[fill=blockBlue]| \alpha^{3}\beta^{3} & |[fill=blockBlue]| \alpha^{4}\beta^{4} \\
|[fill=blockBlue]| \alpha & |[fill=blockBlue]| \alpha^{2}\beta & |[fill=blockBlue]| \alpha^{3}\beta^{2} & |[fill=blockBlue]| \alpha^{4}\beta^{3} & |[fill=blockBlue]| \alpha^{5}\beta^{4} \\
|[fill=blockBlue]| \alpha^{2} & |[fill=blockBlue]| \alpha^{3}\beta & |[fill=blockBlue]| \alpha^{4}\beta^{2} & |[fill=blockBlue]| \alpha^{5}\beta^{3} & |[fill=blockBlue]| \alpha^{6}\beta^{4} \\
|[fill=blockBlue]| \alpha^{3} & |[fill=blockBlue]| \alpha^{4}\beta & |[fill=blockBlue]| \alpha^{5}\beta^{2} & |[fill=blockBlue]| \alpha^{6}\beta^{3} & |[fill=blockBlue]| \alpha^{7}\beta^{4} \\
|[fill=blockBlue]| \alpha^{4} & |[fill=blockBlue]| \alpha^{5}\beta & |[fill=blockBlue]| \alpha^{6}\beta^{2} & |[fill=blockBlue]| \alpha^{7}\beta^{3} & |[fill=blockBlue]| \alpha^{8}\beta^{4} \\
|[fill=blockGreen]| \alpha^{5} & |[fill=blockGreen]| \alpha^{6}\beta & |[fill=blockYellow]| \alpha^{7}\beta^{2} & |[fill=blockYellow]| \alpha^{8}\beta^{3} & |[fill=tinyBlockA]| \alpha^{9}\beta^{4} \\
|[fill=blockGreen]| \alpha^{6} & |[fill=blockGreen]| \alpha^{7}\beta & |[fill=blockYellow]| \alpha^{8}\beta^{2} & |[fill=blockYellow]| \alpha^{9}\beta^{3} & |[fill=tinyBlockB]| \alpha^{10}\beta^{4} \\
};
\node[titleNode, above=1pt] at (m0.north) {Step 1: Euclidean Square Partition of $B$};

\draw[arrowLine] (m0.north east) to[out=20, in=180] (m1.west);
\draw[arrowLine]
  (m1.south east) to[bend left=28]
  (m2.north east);
\draw[arrowLine]
  (m2.south east) to[bend left=10]
  (m3.north east);

\end{tikzpicture}
\end{adjustbox}

\end{minipage}
\caption{The transformation procedure for $(p,s)=(7,5)$.}
\label{fig:boxed-tikz}
\end{figure*}

The columns of~$R$ indexed by $0,1,2,3,4$ coincide with those of~$\overline R$. Each entry in the columns of~$R$ with indices $5$ and $6$ can be expressed as an $\mathbb{F}$-linear combination of entries of~$\overline{R}$, as shown by the following recovery sequence.
\begin{enumerate}
  \item $R(0,5)=\alpha^5=\alpha^{5}(1+\beta^{2}+\beta^{4})-\alpha^{5}\beta^{2}-\alpha^{5}\beta^{4}=\overline{R}(0,5)-\overline{R}(3,2)-\overline{R}(1,4)$,
  \item $R(1,5)=\alpha^6=\alpha^{6}(1+\beta^{2}+\beta^{4})-\alpha^{6}\beta^{2}-\alpha^{6}\beta^{4}=\overline{R}(1,5)-\overline{R}(4,2)-\overline{R}(2,4)$,
  \item $R(2,5)=\alpha^7\beta^2=\alpha^{7}(1+\beta^{2}+\beta^{4})-\alpha^7-\alpha^7\beta^{4}=\overline{R}(2,5)-\alpha^7-\overline{R}(3,4)$,
  \item $R(3,5)=\alpha^8\beta^2=\alpha^{8}(1+\beta^{2}+\beta^{4})-\alpha^{8}-\alpha^{8}\beta^{4}=\overline{R}(3,5)-\alpha^8-\overline{R}(4,4)$, 
  \item $R(4,5)=\alpha^9\beta^4=\alpha^{9}(1+\beta^{2}+\beta^{4})-\alpha^{9}-\alpha^{9}\beta^{2}=\overline{R}(4,5)-\alpha^9-\alpha^9\beta^2$, 
  \item $R(0,6)=\alpha^6\beta=\alpha^{6}(\beta+\beta^{3}+\beta^{4})-\alpha^{6}\beta^{3}-\alpha^{6}\beta^{4}=\overline{R}(0,6)-\overline{R}(3,3)-\overline{R}(2,4)$,
  \item $R(1,6)=\alpha^7\beta=\alpha^{7}(\beta+\beta^{3}+\beta^{4})-\alpha^{7}\beta^{3}-\alpha^{7}\beta^{4}=\overline{R}(1,6)-\overline{R}(4,3)-\overline{R}(3,4)$,
  \item $R(2,6)=\alpha^8\beta^3=\alpha^{8}(\beta+\beta^{3}+\beta^{4})-\alpha^{8}\beta-\alpha^{8}\beta^{4}=\overline{R}(2,6)-\alpha^8\beta-\overline{R}(4,4)$,
  \item $R(3,6)=\alpha^9\beta^3=\alpha^{9}(\beta+\beta^{3}+\beta^{4})-\alpha^{9}\beta-\alpha^{9}\beta^{4}=\overline{R}(3,6)-\alpha^9\beta-R(4,5)$,
  \item $R(4,6)=\alpha^{10}\beta^4=\alpha^{10}(\beta+\beta^{3}+\beta^{4})-\alpha^{10}\beta-\alpha^{10}\beta^{3}=\overline{R}(4,6)-\alpha^{10}\beta-\alpha^{10}\beta^3$.
\end{enumerate}
Thus $R(0,5), R(1,5), R(0,6), R(1,6)$ are already $\F$-linear combinations of entries of $\overline R$. For the remaining terms, note that $[\F(\alpha):\F]=p=7$, hence $\F(\alpha)=\spun_{\F}\{1, \alpha, \cdots, \alpha^6\}$. In particular,
\[
\begin{gathered}
\alpha^7,\alpha^8,\alpha^9\in \spun_{\F}\{1,\alpha,\cdots,\alpha^6\}= \spun_{\F}\bigl(\{\overline{R}(t,0):0\le t\le 4\}\cup\{R(0,5),R(1,5)\}\bigr),\\[0.6ex]
\alpha^9\beta^2\in \spun_{\F}\{\alpha^2\beta^2,\alpha^3\beta^2,\cdots,\alpha^8\beta^2\}= \spun_{\F}\bigl(\{\overline{R}(t,2):0\le t\le 4\}\cup\{R(2,5),R(3,5)\}\bigr),\\[0.6ex]
\alpha^8\beta,\alpha^9\beta,\alpha^{10}\beta\in \spun_{\F}\{\alpha\beta,\alpha^2\beta,\cdots,\alpha^7\beta\}= \spun_{\F}\bigl(\{\overline{R}(t,1):0\le t\le 4\}\cup\{R(0,6),R(1,6)\}\bigr),\\[0.6ex]
\alpha^{10}\beta^3\in \spun_{\F}\{\alpha^3\beta^3,\alpha^4\beta^3,\cdots,\alpha^9\beta^3\}= \spun_{\F}\bigl(\{\overline{R}(t,3):0\le t\le 4\}\cup\{R(2,6),R(3,6)\}\bigr).
\end{gathered}
\]
Therefore every entry of $R$ is an $\F$-linear combination of entries of $\overline R$.
The recovery above illustrates the general mechanism: each interference monomial subtracted from an entry of~$\overline R$ either lies in an already established $\beta^{\theta}$-layer (i.e.\ $\beta^\theta \E$ for some integer $\theta$) or equals an entry recovered at an earlier step. 

Next, we give formal definitions of the basis{-}transformation technique.

\subsection{Formal definition of the basis{-}transformation technique}
We now define the basis-transformation technique, which consists of three operations: Euclidean Square Partition, Transposition, and Column Aggregation.

\paragraph{Euclidean Square Partition}
Partition the array~$B$ into square blocks whose side lengths are determined by the Euclidean algorithm for $(p,s)$.

\begin{definition}[Euclidean Square Partition]\label{def:partition}
Let $h,w$ be positive integers and consider the rectangle
$\mathcal{I}(h,w)=[0,h-1]\times[0,w-1]$.
Let $t:=\min\{h,w\}$. Tile $\mathcal{I}(h,w)$ by placing as many $t\times t$ squares as
possible along the longer side (left-to-right if $w\ge h$, and top-to-bottom if $h>w$).
This leaves a (possibly empty) residual rectangle $\mathcal{I}(h',w')$; if it is nonempty,
apply the same procedure recursively to $\mathcal{I}(h',w')$.
The resulting collection of squares forms a disjoint partition of $\mathcal{I}(h,w)$,
called the \emph{Euclidean Square Partition} of $\mathcal{I}(h,w)$, and is denoted by $\mathcal{P}(h,w)$. Each square in $\mathcal{P}(h,w)$ is called a \emph{Euclidean square}.
\end{definition}

Applying the above procedure to the index set of the basis array $B$ defined as in Fig.~\ref{fig:basis0}, namely for $(h,w)=(p,s)$, yields a partition of $B$ into squares whose side lengths are given by the multiset
\[
\Big\{\underbrace{s,\dots,s}_{a_0},\underbrace{b_1,\dots,b_1}_{a_1},\dots,
\underbrace{b_m,\dots,b_m}_{a_m}\Big\},
\]
as induced by the Euclidean chain~Eq.~\eqref{eq:euclid}.

\paragraph{Transposition}
 We introduce the Transposition operation on $B$, which converts $B$ into an $s\times p$ array $R$ by an overall transpose followed by a block-wise transpose.

\begin{definition}[Transposition]\label{def:Transposition}
Let $B$ be the $p\times s$ array defined as in Fig.~\ref{fig:basis0} and let $B^{\top}$ be the $s\times p$ array defined by
$B^{\top}(i,j)=B(j,i)$.
For each Euclidean square $Q$ in $B$ with top-left corner index $(x,y)$ and side length $t$,
let $Q^{\top}$ be the corresponding $t\times t$ square in $B^{\top}$ with top-left corner index $(y,x)$.
Define $R=\mathsf{Trans}_{p,s}(B)$ as the array obtained from $B^{\top}$ by transposing every square $Q^{\top}$,
equivalently, for every such $Q$ and all $0\le u,v\le t-1$,
\begin{equation}\label{eq:blockwise}
R(y+u,\ x+v)=B(x+u,\ y+v).
\end{equation}
\end{definition}

\paragraph{Column Aggregation}
 For each column of $R$, we keep the $\alpha$ part of each entry unchanged and replace the $\beta$ part by the sum of all distinct $\beta$-powers appearing in that column.

\begin{definition}
Since every entry of~$R$ is a monomial in $\alpha$ and~$\beta$,
for $(i,j)\in \I(s,p)$ we write
\[
R(i,j) = \alpha^{\gamma(i,j)}\,\beta^{\tau(i,j)},
\]
where $\gamma(i,j)$ and $\tau(i,j)$ denote the exponents of $\alpha$ and
$\beta$ in $R(i,j)$, respectively.
\end{definition}

\begin{definition}[Column Aggregation]\label{def:aggre}
For a column~$c$ of~$R$, the \emph{column aggregation}
operator~$\Agg(\cdot)$ replaces each entry $R(i,c)$ by
\[
\Agg(R(i,c))
:= \alpha^{\gamma(i,c)}
   \!\!\sum_{\theta\,\in\,
     \{\tau(i',c)\,:\,i'\in[0,s-1]\}}
   \!\!\beta^{\theta}.
\]
Applying~$\Agg(\cdot)$ to each column of~$R$ yields
the $s\times p$ array~$\overline{R}$.
\end{definition}

\begin{example}
In Fig.~\ref{fig:boxed-tikz} with $(p,s)=(7,5)$, column~$5$ of $R$ consists of entries
\[
\alpha^5,\ \alpha^6,\ \alpha^7\beta^2,\ \alpha^8\beta^2,\ \alpha^9\beta^4.
\]
 Column aggregation replaces each entry by
\[
\alpha^5(1+\beta^2+\beta^4),\ \alpha^6(1+\beta^2+\beta^4),\ \dots,\
\alpha^9(1+\beta^2+\beta^4),
\]
as shown in the bottom array $\overline{R}$ of Fig.~\ref{fig:boxed-tikz}.
\end{example}

Define the subspace
\begin{equation}\label{eq:S-def}
S:=\spun_{\F}\{\overline{R}(0,c):\,0\le c\le p-1\},
\end{equation}
and we verify in Section~\ref{subsec:proof-main} that $S$ satisfies~Eq.~\eqref{eq:smain}.

\subsection{Structure of the transposed array \texorpdfstring{$R$}{R}}\label{subsec:R1-form}

In this subsection, we characterize the properties of the transposed array $R$ defined as in Definition~\ref{def:Transposition} that will be used for the proof of Theorem~\ref{thm:main}.

We begin by recording two basic properties that follow directly from the definitions of $R$ and $B$ and will be used throughout.

\begin{observation}[Structural properties of~$R$]\label{obs:1}
The transposed array~$R$ has the following properties:
\begin{enumerate}
\item[(1)] The Euclidean Square Partition $\mathcal P(p,s)$ of~$B$ induces a Euclidean Square Partition $\mathcal P(s,p)$ of~$R$; the Euclidean square of~$B$ with top-left corner index $(x,y)$ becomes the Euclidean square of~$R$ whose top-left entry is $R(y,x)=B(x,y)=\alpha^x(\alpha\beta)^y$.
\item[(2)] Inside each Euclidean square of~$R$, the entry immediately to the right of any given entry is obtained by multiplying that entry by~$\alpha\beta$, while the entry immediately below it is obtained by multiplying that entry by~$\alpha$.
\end{enumerate}
\end{observation}
The preceding observations allow us to determine the $\alpha$ and $\beta$-exponents of each entry of~$R$.

\begin{lemma}\label{lem:Rformulas}
Let $(i,j)\in \I(s,p)$ and $Q=[\rho,\rho+\lambda-1]\times[\kappa,\kappa+\lambda-1]$
    be the Euclidean square of $\mathcal P(s,p)$ containing $(i,j)$. Then 
\begin{enumerate}
    \item[(1)] $\gamma(i,j)=i+j$.

    \item[(2)] $\tau(i,j)=\rho+(j-\kappa).$
\end{enumerate}
In particular, $\tau(i,j)\in [0,s-1]$.
\end{lemma}

\begin{proof}
Assume that
\[
i=\rho+r,\qquad j=\kappa+c,
\]
where $0\le r,c\le \lambda-1$. By Observation~\ref{obs:1}(1),
\[
R(\rho,\kappa)=\alpha^{\rho+\kappa}\beta^\rho.
\]
By Observation~\ref{obs:1}(2),
\[
R(i,j)=R(\rho+r,\kappa+c)=\alpha^{\rho+\kappa}\beta^\rho\cdot \alpha^r(\alpha\beta)^c=\alpha^{\rho+\kappa+r+c}\beta^{\rho+c}=\alpha^{i+j}\beta^{\rho+(j-\kappa)}.
\]
Hence
\[
\gamma(i,j)=i+j,
\qquad
\tau(i,j)=\rho+(j-\kappa).
\]
Finally, since $0\le \rho\le s-\lambda$ and $0\le j-\kappa=c\le \lambda-1$, we obtain
\[
0\le \tau(i,j)=\rho+(j-\kappa)\le (s-\lambda)+(\lambda-1)=s-1.
\]
This completes the proof.
\end{proof}

From Lemma~\ref{lem:Rformulas}, we obtain the following proposition.

\begin{proposition}\label{prop:beta-structure}
Let $R$ be the transposed array defined as in Definition~\ref{def:Transposition}. Then the following statements hold.
    \begin{enumerate}
        \item[(1)] The $\alpha$-exponent is consecutive along every column of $R$.
        \item[(2)] The $\beta$-exponent is constant along every column of each Euclidean square of $R$.
        \item[(3)] The operator $\Agg(\cdot)$ acts as the identity on the first $a_0s$ columns of $R$.
        \item[(4)] In each column of $R$, entries in lower squares, if any, have larger $\beta$-exponents.

    \end{enumerate}
\end{proposition}

\begin{proof}
(1) By Lemma~\ref{lem:Rformulas} (1), for any $(i,j)\in \I(s,p)$, the $\alpha$-exponent of $R(i,j)$ is $\gamma(i,j)=i+j$. Therefore, within any fixed column $j$, the $\alpha$-exponents are consecutive.

(2)
Let
\[
Q=[\rho,\rho+\lambda-1]\times[\kappa,\kappa+\lambda-1]
\]
be a Euclidean square of $\mathcal P(s,p)$. By Lemma~\ref{lem:Rformulas}(2), for every $(i,j)\in Q$,
\[
\tau(i,j)=\rho+(j-\kappa).
\]
Hence $\tau(i,j)$ depends only on the column index $j$, so the $\beta$-exponent is constant along every column of $Q$.

(3) Each of the first $a_0s$ columns of $R$ is contained in a single $s\times s$ block induced by $\mathcal P(p,s)$. By Part~(2), the $\beta$-exponent is constant along every such column. Hence the set of $\beta$-components is a singleton for each of these columns, so $\Agg(\cdot)$ acts as the identity on them.

(4) By Part~(3), it is enough to consider the last $b_1$ columns of $R$. Fix a column $j\in[a_0s,p-1]$. Let
\[
Q_{\mathrm{up}}=[\rho_{\mathrm{up}},\rho_{\mathrm{up}}+\lambda_{\mathrm{up}}-1]\times[\kappa_{\mathrm{up}},\kappa_{\mathrm{up}}+\lambda_{\mathrm{up}}-1]
\]
and
\[
Q_{\mathrm{low}}=[\rho_{\mathrm{low}},\rho_{\mathrm{low}}+\lambda_{\mathrm{low}}-1]\times[\kappa_{\mathrm{low}},\kappa_{\mathrm{low}}+\lambda_{\mathrm{low}}-1]
\]
be two squares of $\mathcal P(s,p)$ crossed by column~$j$, with $Q_{\mathrm{low}}$ strictly below $Q_{\mathrm{up}}$. By Lemma~\ref{lem:Rformulas} (2), the $\beta$-exponent {of the entries} in column $j$ {of $R$ lying in} $Q_{\mathrm{up}}$ is at most
\[
\rho_{\mathrm{up}}+\lambda_{\mathrm{up}}-1,
\]
while the $\beta$-exponent {of the entries in} column $j$ {of $R$ lying} in $Q_{\mathrm{low}}$ is at least
\[
\rho_{\mathrm{low}}.
\]
Since $Q_{\mathrm{low}}$ lies strictly below $Q_{\mathrm{up}}$, we have
\[
\rho_{\mathrm{up}}+\lambda_{\mathrm{up}}-1<\rho_{\mathrm{low}}.
\]
Therefore, in any column $j${ of $R$}, the $\beta$-exponent {of the entries in} $Q_{\mathrm{low}}$ is strictly larger than that {of the entries in} $Q_{\mathrm{up}}$.
\end{proof}

\subsection{Proof of Theorem~\ref{thm:main}}\label{subsec:proof-main}

Recall that, by Eq.~\eqref{eq:S-def}, $S$ is the $\mathbb{F}$-linear span of the first row of
$\overline R$, where $R$ denotes the transposed array and $\overline R$
denotes the array obtained from $R$ by column aggregation. The explicit
description of $S$ stated in Theorem~\ref{thm:main} is derived in Appendix~\ref{app:explicit-basis}. In this
subsection, we prove the first assertion of Theorem~\ref{thm:main}, namely,
\begin{equation*}
 \dim_{\F} S=p, \qquad \K=\sum_{u=0}^{s-1}{S}\alpha^u .
\end{equation*}

Since $\Agg(\cdot)$ does not change the $\alpha$-exponents, Proposition~\ref{prop:beta-structure} (1) implies that the entries in each column of $\overline R$ have consecutive $\alpha$-powers. Hence, the $u$-th row of $\overline R$ is equal to $\alpha^u$ times the first row, i.e. $\overline R(u,v)=\alpha^u\overline R(0,v)$ for any  $(u,v)\in\I(s,p)$. Set
\begin{equation}\label{eq:K}
  K:=\sum_{u=0}^{s-1}{S}\alpha^u
=\spun_{\F}\{\overline R(u,v):(u,v)\in\I(s,p)\}.  
\end{equation}

We first prove that $K=\K$.  Note that $K \subseteq \K$ since every entry of $\overline{R}$ is a sum of monomials in $\K$. Moreover, since $R$ is a rearrangement of the basis array $B$, we have
\[
\spun_{\F}(R)=\spun_{\F}(B)=\K.
\]
Therefore, to prove $\K\subseteq K$, it is enough to show that every entry of $R$ lies in $K$. 

By Proposition~\ref{prop:beta-structure}(3), the operator $\Agg(\cdot)$ acts as the identity on the first $a_0s$ columns of $R$. Hence, all entries in these columns already belong
to $K$. {If no columns remain, then $K=\K$; otherwise,} {it} remains to consider the last $b_1$ columns.
Define
\begin{equation*}
R_1:=\bigl(R(i,a_0s+j)\bigr)_{(i,j)\in\I(s,b_1)},
\end{equation*}
that is, $R_1$ is the $s\times b_1$ subarray formed by the last $b_1$ columns of~$R$. We use $(i,j)\in\I(s,b_1)$ as local coordinates on~$R_1$ via the identification
\[
R_1(i,j)\longleftrightarrow R(i,a_0s+j).
\]
For $j\in[0,b_1-1]$, define
\[
\mathsf{Exp}(j)\!:=\!\{t\in[0,s-1]:\beta^t \text{ appears in column }a_0s+j\text{ of }R\},
\]
and, for $(i,j)\in\I(s,b_1)$, set
\[
t_{i,j}:=\tau(i,a_0s+j).
\]
Thus $\mathsf{Exp}(j)$ is the set of $\beta$-exponents occurring in the local column $j$, and $t_{i,j}$ is the $\beta$-exponent of the target entry $R(i,a_0s+j)$. We shall prove
\begin{equation}\label{eq:goal}
R(i,a_0s+j)\in K, \qquad(i,j)\in\I(s,b_1).
\end{equation}

For every $(i,j)\in\I(s,b_1)$, column aggregation gives
\begin{equation}\label{eq:recover}
R(i,a_0s+j)
=
\overline R(i,a_0s+j)
-
\alpha^{a_0s+i+j}
\sum_{t\in\mathsf{Exp}(j)\setminus\{t_{i,j}\}}\beta^t,
\end{equation}
and,  by Eq.~\eqref{eq:K}, we have 
\begin{equation}\label{eq:AggSum}
\overline R(i,a_0s+j)\in K.
\end{equation}
Therefore, to prove Eq.~\eqref{eq:goal}, it suffices to show that the interference monomial
\[
 M_t(i,j):=\alpha^{a_0s+i+j}\beta^t\in K
\]
for every $t\in\mathsf{Exp}(j)\setminus\{t_{i,j}\}$.

We split the interference exponents
\(t\in\mathsf{Exp}(j)\setminus\{t_{i,j}\}\)
into the following three cases:
\[
\begin{array}{ll}
\mathrm{(C1)} & i+j<t,\\[1mm]
\mathrm{(C2)} & i+j\ge t \text{ and } t<t_{i,j},\\[1mm]
\mathrm{(C3)} & i+j\ge t \text{ and } t>t_{i,j}.
\end{array}
\]
These cases are pairwise disjoint and exhaustive. Indeed, if (C1) does not hold, then
\(i+j\ge t\). Since \(t\ne t_{i,j}\), either \(t<t_{i,j}\) or \(t>t_{i,j}\).

The three cases are handled by different mechanisms. Case~(C1) is direct.
Case~(C2) will be handled by proving that the whole \(\E\beta^t\)-layer is
contained in \(K\). Case~(C3) will be handled by proving that the required
monomial \(M_t(i,j)\) is obtained from a predecessor with respect to an order
to be defined on \(\I(s,b_1)\).

We also record that the condition \(i+j\ge t\) in Case (C2) is automatic. 
Fix \(t\in\mathsf{Exp}(j)\setminus\{t_{i,j}\}\) with \(t<t_{i,j}\). Let
\[
Q=[\rho,\rho+\lambda-1]\times[\kappa,\kappa+\lambda-1]
\]
be the Euclidean square of \(\mathcal{P}(s,b_1)\) whose entries lying in the local column \(j\) have \(\beta\)-exponent \(t\). Then,  Lemma~\ref{lem:Rformulas}(2) implies
\[
t=\rho+(j-\kappa).
\]
By Proposition~\ref{prop:beta-structure}(4), the unique Euclidean square containing
\((i,j)\) lies strictly below \(Q\). Hence \(i\ge \rho+\lambda\). Since
\(j\ge\kappa\), we have
\[
i+j-t=i-\rho+\kappa\ge \lambda+\kappa\ge 1.
\]
Thus \(i+j>t\). Therefore, in the proof of Case (C2), it suffices to use
the condition \(t<t_{i,j}\).

The following lemma supplies the direct argument for Case~(C1).

\begin{lemma}\label{lem:initial}
For all $0\le u\le a_0s-1$ and $0\le \theta\le s-1$,
\[
\alpha^u(\alpha\beta)^\theta\in K.
\]
\end{lemma}

\begin{proof}

By Proposition~\ref{prop:beta-structure}(3) and Lemma~\ref{lem:Rformulas}, for $0\le r\le a_0-1$ and $0\le \theta\le s-1$, we have
\[
\overline R(0,rs+\theta)=R(0,rs+\theta)=\alpha^{rs}(\alpha\beta)^\theta\in S\subseteq K.
\]
Multiplying by $\alpha^v$, $0\le v\le s-1$, Eq.~\eqref{eq:K} implies
\[
\alpha^{rs+v}(\alpha\beta)^\theta\in K,
\]
which means $\alpha^u(\alpha\beta)^\theta\in K$ for all $0\le u\le a_0s-1$ and $0\le \theta\le s-1$.
\end{proof}

We now apply Lemma~\ref{lem:initial} to Case~(C1). 
If $i+j<t$, then
\[
 {1}\le a_0s+i+j-t\le a_0s-1.
\]
Applying Lemma~\ref{lem:initial} with $u=a_0s+i+j-t$ and $\theta=t$, we obtain
\begin{equation}\label{eq:c1}
 M_t(i,j)
 =\alpha^{a_0s+i+j}\beta^t
 =\alpha^{a_0s+i+j-t}(\alpha\beta)^t
 \in K .                                          
\end{equation}

The next lemma is used for the proof of Case (C2). In this case, we recover not only the monomial $M_t(i,j)$,
but the entire $\E\beta^t$-layer.

\begin{lemma}\label{lem:field-gen}
Let $\theta\in[0,s-1]$. If there exists $z\in\mathbb Z$ such that
\[
\alpha^z\beta^\theta,\ \alpha^{z+1}\beta^\theta,\ \dots,\ \alpha^{z+p-1}\beta^\theta\in K,
\]
then $\E\beta^\theta\subseteq K$. 
\end{lemma}

\begin{proof}
Since $[\E:\F]=p$ and $\alpha$ generates $\E$ over $\F$, the $p$ consecutive powers
\[
\alpha^z,\alpha^{z+1},\dots,\alpha^{z+p-1}
\]
form an $\F$-basis of $\E$. Hence
\[
\spun_{\F}\{\alpha^z\beta^\theta,\alpha^{z+1}\beta^\theta,\dots,\alpha^{z+p-1}\beta^\theta\}
=
\E\beta^\theta
\subseteq K.
\]
\end{proof}

The remaining arguments take place inside $R_1$. We first identify the
partition induced on this subarray.

\begin{lemma}[Partition of $R_1$]\label{lem:R1-partition}
The Euclidean squares induced on the subarray $R_1$ are exactly the Euclidean Square Partition $\mathcal P(s,b_1)$ on its local index set $\I(s,b_1)$.
\end{lemma}

\begin{proof}
The transpose $(x,y)\mapsto (y,x)$ maps $\mathcal P(p,s)$ onto $\mathcal P(s,p)$. The Transposition operation transposes each square without changing the partition, so $R$ is partitioned by $\mathcal P(s,p)$. Removing the first $a_0s$ columns leaves the rectangle $[0,s-1]\times[a_0s,p-1]$, namely the subarray $R_1$; after reindexing its columns by $j\mapsto j-a_0s$, its local index set becomes $\I(s,b_1)=[0,s-1]\times[0,b_1-1]$. By Definition~\ref{def:partition}, the induced partition on this local index set is exactly $\mathcal P(s,b_1)$.
\end{proof}

We next introduce the locating map.

\begin{definition}[Locating map]\label{def:locating-map}
Given $\theta\in[0,s-1]$ and $\xi\in[0,b_1-1]$, define
\[
\operatorname{Loc}(\theta, \xi):=(i_{\theta,\xi},j_{\theta,\xi})\in \mathcal{I}(s,b_1)
\]
as the local coordinates of the entry of $R_1$ such that $R(i_{\theta,\xi},a_0s+j_{\theta,\xi})=\alpha^{a_0s+\theta+\xi}\beta^\theta$.

\end{definition}
By~Eq.~\eqref{eq:B-def}, $\alpha^{a_0s+\theta+\xi}\beta^\theta = B(a_0s+\xi,\,\theta)$,
which lies in row $a_0s+\xi\ge a_0s$ of $B$, hence in the last $b_1$ rows.
By~Eq.~\eqref{eq:blockwise}, every entry in rows $[a_0s,p-1]$ of $B$ is
mapped by the Transposition to a column $\ge a_0s$ in $R$, so this
monomial appears as an entry of $R_1$.
Since $R$ is a rearrangement of the $\mathbb{F}$-basis $B$, its entries
are pairwise distinct, so the position is unique.
Hence $\operatorname{Loc}(\theta,\xi)$ is well-defined.

\begin{lemma}\label{lem:layer-segment}
Given $\theta\in[0,s-1]$ and $\xi\in[0,b_1-1]$,
let
\[
Q_{\theta,\xi}=[\rho,\rho+\lambda-1]\times[\kappa,\kappa+\lambda-1]\in \mathcal{P}(s,b_1)
\]
be the unique square of $\mathcal P(s,b_1)$ containing $(\theta,\xi)${, whose existence and uniqueness follow from} Lemma~\ref{lem:R1-partition}. 
Let $\operatorname{Loc}(\theta, \xi)
=(i_{\theta,\xi},j_{\theta,\xi})$ be defined as in Definition~\ref{def:locating-map}. 
Then
\begin{equation}\label{eq:locsquare}
(i_{\theta,\xi},j_{\theta,\xi})
=
\bigl(\rho+(\xi-\kappa),\ \kappa+(\theta-\rho)\bigr)\in Q_{\theta,\xi}.
\end{equation}

\end{lemma}

\begin{proof}

By Lemma~\ref{lem:Rformulas}(1),
\[
\gamma(\rho+(\xi-\kappa),a_0s+\kappa+(\theta-\rho))
=
a_0s+\theta+\xi.
\]
By Lemma~\ref{lem:Rformulas}(2),
\[
\tau(\rho+(\xi-\kappa),a_0s+\kappa+(\theta-\rho)) =\rho+\bigl(\kappa+(\theta-\rho)-\kappa\bigr)=\rho+(\theta-\rho)=\theta.
\]
Therefore,
\[
R(i_{\theta,\xi},a_0s+j_{\theta,\xi}) =\alpha^{a_0s+\theta+\xi}\beta^\theta=R(\rho+(\xi-\kappa),a_0s+\kappa+(\theta-\rho)).
\]
Since all entries of~$R_1$ are pairwise distinct, we have $(i_{\theta,\xi},j_{\theta,\xi})
=
\bigl(\rho+(\xi-\kappa),\ \kappa+(\theta-\rho)\bigr)$.

Moreover, as $(\theta,\xi)\in Q_{\theta,\xi}$ and both $\theta-\rho$ and $\xi-\kappa$ belong to $[0,\lambda-1]$, it follows that $(i_{\theta,\xi},j_{\theta,\xi})\in Q_{\theta,\xi}$. 
\end{proof}

Cases~(C2) and~(C3) will require induction. The induction is taken with respect to the column-major lexicographic order on $\I(s,b_1)$, defined by
\[
(i',j')\prec(i,j)
\quad\Longleftrightarrow\quad
j'<j,\ \text{or}\ j'=j\ \text{and}\ i'<i.
\]
Both indices are bounded, so $\prec$ is a well-order on $\I(s,b_1)$.

The next two lemmas are crucial for the case analysis in the inductive step. Lemma~\ref{lem:pred-preceding-fixed} handles Case~(C2) $t<t_{i,j}$, and Lemma~\ref{lem:subdiag} handles Case~(C3):
$t>t_{i,j}$ with $i+j\ge t$. Their proofs are deferred to
Appendix~\ref{app:proof-lem-main}.

\begin{lemma}\label{lem:pred-preceding-fixed}
Fix $(i,j)\in\I(s,b_1)$. If $t\in\mathsf{Exp}(j)$ satisfies $t<t_{i,j}$, then for every $\xi\in[0,b_1-1]$,
\[
\operatorname{Loc}(t, \xi)=(i_{t,\xi},j_{t,\xi})\prec(i,j).
\]
\end{lemma}

\begin{lemma}\label{lem:subdiag}
Fix \((i,j)\in\I(s,b_1)\). 
If \(t\in\mathsf{Exp}(j)\) satisfies \(t>t_{i,j}\) and \(i+j\ge t\), and if we set
\[
\delta:=i+j-t,
\]
then
\[
0\le \delta\le b_1-1,
\]
and
\[
\operatorname{Loc}(t,\delta)=(i_{t,\delta},j_{t,\delta})\prec(i,j).
\]
\end{lemma}

We now prove Eq.~\eqref{eq:goal} by induction on $\prec$.  


\paragraph{Initial segment} 
We first prove Eq.~\eqref{eq:goal} for
\[
 (i,0),\qquad 0\le i\le b_1-1.
\]
Each such point lies in the top $b_1\times b_1$ square of $R_1$,
whose corresponding index set in $R$ is
\[
[0,b_1-1]\times [a_0s,p-1].
\]
Hence Lemma~\ref{lem:Rformulas}(2) gives
\[
t_{i,0}=\tau(i,a_0s)=0+(a_0s-a_0s)=0.
\]
Every exponent
\[
t\in \mathsf{Exp}(0)\setminus\{0\}
\]
{appears in the local column \(0\) of \(R_1\) within} a Euclidean square of \(\mathcal P(s,b_1)\) that is
strictly below the top \(b_1\times b_1\) square. Let this square be
\[
Q=[\rho,\rho+\lambda-1]\times[\kappa,\kappa+\lambda-1].
\]
Since \(Q\) is crossed by the local column \(0\), we must have \(\kappa=0\).
Moreover, because \(Q\) is strictly below the top square, \(\rho\ge b_1\).
By Lemma~\ref{lem:Rformulas}(2), the exponent {of the entries in} \(Q\) {lying in} the local column \(0\) {of \(R_1\)} is
\[
t=\rho+(a_0s+0)-(a_0s+\kappa)=\rho\ge b_1.
\]
Thus, for \(0\le i\le b_1-1\), we have
\[
i+0\le b_1-1<t.
\]
Therefore Case (C1) applies. 
By Eq.~\eqref{eq:c1}, we have $M_t(i,0)\in K$, and
Eq.~\eqref{eq:recover} then shows that the corresponding entry
$R(i,a_0s)$ belongs to $K$.

\paragraph{Inductive step}
Now fix $(i,j)\in\I(s,b_1)$, and assume that
\[
R(i',a_0s+j')\in K
\qquad
\text{for all }(i',j')\prec(i,j).
\]
We show that $R(i,a_0s+j)\in K$. By Eqs.~\eqref{eq:recover} and \eqref{eq:AggSum}, it remains to prove that
\[
 M_t(i,j)=\alpha^{a_0s+i+j}\beta^t\in K
\]
for every $t\in\mathsf{Exp}(j)\setminus\{t_{i,j}\}$.

We distinguish the three cases for $t\in\mathsf{Exp}(j)\setminus\{t_{i,j}\}$.
\begin{enumerate}
\item[(C1)]  $i+j<t$. By Eq.~\eqref{eq:c1}, $M_t(i,j)\in K$.

\item[(C2)] $i+j\ge t$ and $t<t_{i,j}$. Lemma~\ref{lem:initial} gives
\[
\alpha^t\beta^t,\ \alpha^{t+1}\beta^t,\ \dots,\ \alpha^{t+a_0s-1}\beta^t\in K.
\]
For each $\xi\in[0,b_1-1]$, by Lemma~\ref{lem:pred-preceding-fixed}, we have $(i_{t,\xi},j_{t,\xi})\prec(i,j)$, so the induction hypothesis gives $R(i_{t,\xi},a_0s+j_{t,\xi})\in K$. Lemma~\ref{lem:layer-segment} then yields
\[
R(i_{t,\xi},a_0s+j_{t,\xi}) = \alpha^{a_0s+t+\xi}\beta^t \in K,\qquad 0\le \xi\le b_1-1.
\]
Therefore $K$ contains the $p=a_0s+b_1$ consecutive elements
\[
\alpha^t\beta^t,\ \alpha^{t+1}\beta^t,\ \dots,\ \alpha^{t+p-1}\beta^t.
\]
Lemma~\ref{lem:field-gen} implies $\E\beta^t\subseteq K$. Since $\alpha^{a_0s+i+j}\in\E$, we obtain
\[
 M_t(i,j)=\alpha^{a_0s+i+j}\beta^t\in K.
\]

\item[(C3)] $i+j\ge t$ and $t>t_{i,j}$. Set $\delta:=i+j-t$. By Lemma~\ref{lem:subdiag}, we have
\[
\delta\in[0,b_1-1] \qquad\text{and}\qquad (i_{t,\delta},j_{t,\delta}):=\operatorname{Loc}(t,\delta)\prec(i,j).
\]
By the induction hypothesis, $R(i_{t,\delta},a_0s+j_{t,\delta})\in K$. Lemma~\ref{lem:layer-segment} gives
\[
R(i_{t,\delta},a_0s+j_{t,\delta})
=
\alpha^{a_0s+t+\delta}\beta^t
=
\alpha^{a_0s+i+j}\beta^t.
\]
Hence $ M_t(i,j)=\alpha^{a_0s+i+j}\beta^t\in K$.
\end{enumerate}

We have proved that every term in the summation on the right-hand side of Eq.~\eqref{eq:recover} belongs to $K$. Since $\overline R(i,a_0s+j)\in K$ by Eq.~\eqref{eq:AggSum}, it follows that $R(i,a_0s+j)\in K$. The induction is complete, so Eq.~\eqref{eq:goal} holds for all $(i,j)\in\I(s,b_1)$. Combined with the fact that the first $a_0s$ columns of $R$ are already in $K$ (Proposition~\ref{prop:beta-structure}(3)), every entry of $R$ lies in $K$. Hence
\[
\K=\spun_{\F}(R)\subseteq K.
\]
Because $K\subseteq \K$ by definition, we obtain $K=\K$. Thus
\[
\K=\sum_{u=0}^{s-1}{S}\alpha^u.
\]

Finally, $S$ is spanned by the $p$ elements $\overline R(0,0),\dots,\overline R(0,p-1)$, so $\dim_{\F}S\le p$. On the other hand,
\[
sp=\dim_{\F}\K
=\dim_{\F}K
\le
\sum_{u=0}^{s-1}\dim_{\F}({S}\alpha^u)
=
s\dim_{\F}S
\le
sp.
\]
Hence $\dim_{\F}S=p$. This finishes the proof of Eq.~\eqref{eq:smain}.

\subsection{Application to the optimal repair of Reed-Solomon codes}

Theorem~\ref{thm:main} supplies the optimal repair-enabling subspaces needed in the TYB construction without the congruence condition $p_i\equiv1\pmod{s}$.
\begin{theorem}\label{thm2}
Let $k,n,d$ be positive integers such that $k< d \le n-1$, and set $s:=d+1-k$.
Given a finite field $\Fq$, let $p_1,\dots,p_n$ be any distinct primes with $p_i\ge s$ for all~$i$.
Choose elements $\alpha_1,\dots,\alpha_n$ such that $[\Fq(\alpha_i):\Fq]=p_i$ for $i=1,\dots,n$, and set
\[
\E:=\Fq(\alpha_1,\dots,\alpha_n).
\]
Let $\K$ be an extension field of $\E$ such that $[\K:\E]=s$, and let $\Omega=\{\alpha_1,\dots,\alpha_n\}\subseteq \E$. Then the code $\C:=\mathrm{RS}_{\K}(n,k,\Omega)$, viewed as an array code over $\Fq$ via the extension $\K/\Fq$, is an $(n,k)$ MSR code with repair degree $d$ and subpacketization
$\ell=s\prod_{i=1}^{n}p_i$.

\end{theorem}
\begin{proof}
Fix a failed node $i\in[n]$ and a helper set $\mathcal H\subseteq[n]\setminus\{i\}$ with $|\mathcal H|=d$. By the repair argument in Section~\ref{subsec:tyb17}, it is enough to construct an $F_i$-subspace $S_i\subseteq\K$ satisfying Eq.~\eqref{eq:S}, where $F_i=\Fq(\{\alpha_j:j\neq i\})$. Applying Theorem~\ref{thm:main} to the tower $F_i\subset \E\subset\K$ with $\alpha=\alpha_i$ (noting $[\E:F_i]=p_i\ge s=[\K:\E]$), we obtain an $F_i$-subspace $S_i\subseteq\K$ with
\[
\dim_{F_i}(S_i)=p_i,\qquad
\K=\sum_{u=0}^{s-1}{S_i}\alpha_i^u.
\]
Substituting $S_i$ into the repair argument, each helper node in~$\mathcal H$ transmits $p_i$ symbols over~$F_i$. Since $[F_i:\Fq]=\prod_{j\neq i}p_j$, the total repair bandwidth measured in $\Fq$-symbols is
\[
d\cdot p_i\cdot[F_i:\Fq]
=d\prod_{j=1}^n p_j
=\frac{d}{s}\cdot s\prod_{j=1}^n p_j
=\frac{d\ell}{d-k+1},
\]
which meets the cut-set bound~Eq.~\eqref{eq:cutset}.
\end{proof}
\begin{remark}
To minimize the subpacketization, one takes $p_1,\dots,p_n$ to be the $n$ smallest primes {at least}~$s$.
\end{remark}

\begin{proposition}[Asymptotic comparison with~\cite{TYB17} for fixed $s$]\label{prop:prime-AP}
Fix $s\ge 2$. Let $p_i$ denote the $i$-th smallest prime {at least} $s$, and let $p_i^{\prime}$ denote the $i$-th smallest prime satisfying $p_i^{\prime} \equiv 1 \pmod{s}$. Then, as $n\to\infty$,
\begin{equation}\label{eq:ratio}
  \frac{\ell_{\mathrm{TYB}}}{\ell_{\mathrm{new}}}
  =\frac{\prod_{i=1}^{n}  p_i^{\prime}}{\prod_{i=1}^{n}  p_i}=\varphi(s)^n\,e^{\mathrm{o}(n)}.
\end{equation}
\end{proposition}

\begin{proof}
For fixed $s$, the prime number theorem and its arithmetic-progression version~\cite{IK04} give
\[
p_n \sim n\log n,
\qquad
p_n' \sim \varphi(s)\, n\log n
\qquad (n\to\infty).
\]
Hence
\[
\frac{p_n'}{p_n}\to \varphi(s)
\qquad (n\to\infty),
\]
and therefore
\[
\log\frac{p_n'}{p_n}\to \log\varphi(s).
\]
By the Ces\`aro mean theorem,
\[
\frac{1}{n}\sum_{i=1}^{n}\log\frac{p_i'}{p_i}\to \log\varphi(s).
\]
Equivalently,
\[
\log\frac{\ell_{\mathrm{TYB}}}{\ell_{\mathrm{new}}}
=
\sum_{i=1}^{n}\log\frac{p_i'}{p_i}
=
n\log\varphi(s)+o(n).
\]
Exponentiating proves Eq.~\eqref{eq:ratio}.
\end{proof}

\begin{example}\label{ex:phi2-n12}
For $n=12$ and $s\in\{3,4,6\}$, namely for all admissible $s$ such that $\varphi(s)=2$, the subpacketization of our construction is strictly smaller than that of~\cite{TYB17}. More precisely,
\[
\renewcommand{\arraystretch}{1.15}
\begin{array}{c|c|c|c}
s & \ell_{\mathrm{new}} & \ell_{\mathrm{TYB}} & \ell_{\mathrm{TYB}}/\ell_{\mathrm{new}} \\
\hline
3 & {4.564}\times 10^{{14}} & 6.024\times 10^{19} & {1.320}\times 10^{{5}} \\
4 & 8.722\times 10^{15} & 4.001\times 10^{19} & 4.588\times 10^{3} \\
6 & 1.230\times 10^{17} & 1.205\times 10^{20} & 9.797\times 10^{2}
\end{array}
\]
\end{example}

\section{Conclusion}\label{sec:conclusion}
In this paper, we formulated the construction of repair-enabling subspaces
for optimal repair of Reed-Solomon codes as a basis-transformation problem
over a tower of finite fields. We developed an explicit transformation procedure, consisting of Euclidean Square Partition, Transposition, and Column Aggregation, which establishes the existence of optimal repair-enabling subspaces under the single condition $p\ge s$. This decouples the linear repair mechanism from the arithmetic constraints on the extension degrees used in previous constructions, yielding RS-MSR codes with subpacketization
$\ell=s\prod_{i=1}^n p_i$ for arbitrary distinct primes $p_i\ge s$.

Several questions remain open. The exponential lower bound of Alrabiah and Guruswami~\cite{omar21tit} shows that some exponential dependence of the subpacketization on $n$ is unavoidable for MSR codes. However, a gap remains between this lower bound and the subpacketization of our construction. Closing this gap is an interesting direction. It is also of interest to determine whether the proposed basis-transformation framework extends to multi-node or cooperative repair models, and whether it can support optimal-access repair in these more general settings.

\appendices

\section{Explicit form of the generators of S}\label{app:explicit-basis}

In this appendix, we derive the explicit formula for the generators of $S$ as shown in Eq.~\eqref{eq:S-def}, i.e., $\overline{R}(0,c)$ for $0\le c\le p-1$.

First, by Proposition~\ref{prop:beta-structure}(3), the operator $\Agg(\cdot)$ acts as the identity on the first $a_0s$ columns of $R$. Hence
\[
\overline R(0,rs+u)=R(0,rs+u)=\alpha^{rs}(\alpha\beta)^u,\qquad 0\le r\le a_0-1,\ 0\le u\le s-1,
\]
which are exactly the generators of $S^{(1)}$. {For each remaining column $j\in[0,b_1-1]$, we} determine
\[
\R_j=\overline R(0,a_0s+j),\qquad 0\le j\le b_1-1.
\]
Recall that $R_1$ is the $s\times b_1$ subarray formed by the last $b_1$ columns of $R$. We use local coordinates $(i,j)\in\I(s,b_1)$ on $R_1$ with
\[
R_1(i,j)\longleftrightarrow R(i,a_0s+j).
\]
For each local column $j\in[0,b_1-1]$, the set
$\Exp(j)$ of $\beta$-exponents appearing in that column
is given by
\[
\Exp(j)\!:=\!
\{t\in[0,s-1]:\beta^t\text{ appears in column }a_0s+j\text{ of }R\}.
\]
Therefore
\begin{equation}\label{eq:Rj-exp}
\R_j
=
\alpha^{a_0s+j}\sum_{t\in\Exp(j)}\beta^t .
\end{equation}
Thus the problem reduces to determining \(\Exp(j)\) for each local column \(j\).
By Proposition~\ref{prop:beta-structure}(2), within each Euclidean square of \(\mathcal{P}(s,b_1)\), the \(\beta\)-exponent is constant along each local column. Thus, for a fixed local column \(j\), every square crossed by that column contributes {the common} \(\beta\)-exponent {of its entries in column \(j\) of \(R_1\)} to \(\Exp(j)\).

Let
\[
Q=[\rho,\rho+\lambda-1]\times[\kappa,\kappa+\lambda-1]
\in\mathcal{P}(s,b_1)
\]
be a Euclidean square. We say that \(Q\) \emph{crosses} the local column \(j\) if
\[
\kappa\le j\le \kappa+\lambda-1.
\]
In this case, the \(j\)-column contribution of \(Q\) is defined as
\begin{equation*}
\ctr_j(Q):=\rho+j-\kappa .
\end{equation*}
Indeed, by Lemma~\ref{lem:Rformulas}(2), every entry of \(R_1\) lying in the intersection of \(Q\) with the local column \(j\) has \(\beta\)-exponent 
\begin{equation*}
\rho+(a_0s+j)-(a_0s+\kappa)=\rho+j-\kappa.
\end{equation*}

Hence
\begin{equation}\label{eq:Exp-ctr}
\Exp(j)=\{\ctr_j(Q):Q\in\mathcal P(s,b_1)\text{ crosses column }j\}.
\end{equation}

Before proceeding, we recall the following notations from the Euclidean algorithm that will be frequently
used in the proofs below.
\begin{itemize}
    \item $\nu:=\lfloor m/2\rfloor,\nu':=\lceil m/2\rceil .$
    \item $C_0 := 0$ and $C_r := \sum_{\mu=1}^{r} a_{2\mu} b_{2\mu}$ for $1\le r\le \nu$.
     \item $Y_0 := 0$ and $Y_h := \sum_{\gamma=1}^{h} a_{2\gamma-1} b_{2\gamma-1}$ for $ 1\le h\le \nu'.$
   \item   For \(j\in[0,b_1-1]\),
\[
j^{(0)}:=j,\qquad j^{(2r)}:=j^{(2r-2)} \pmod{b_{2r}}, \qquad 1\le r\le \nu.
\]  
\end{itemize}

Next, we divide the squares of $\mathcal P(s,b_1)$ into two types of layers according to their directions. 

\begin{definition}[Layers and column contributions of \(\mathcal{P}(s,b_1)\)]\label{def:layers-contributions}
For \(1\le h\le \nu'\), the \(h\)-th vertical layer is the family
\[
\mathcal V_h=\{V_{h,w}:0\le w\le a_{2h-1}-1\},
\]
where
\[
V_{h,w}= [Y_{h-1}+wb_{2h-1},\,Y_{h-1}+(w+1)b_{2h-1}-1]\times [C_{h-1},\,b_1-1].
\]
Thus the squares in \(\mathcal V_h\) have side length \(b_{2h-1}\) and are placed from top to bottom.

For \(1\le h\le \nu\), the \(h\)-th horizontal layer is the family
\[
\mathcal H_h=\{H_{h,w}:0\le w\le a_{2h}-1\},
\]
where
\[
H_{h,w}
=
[Y_h,\,s-1]
\times
[C_{h-1}+wb_{2h},\,C_{h-1}+(w+1)b_{2h}-1].
\]
Thus the squares in \(\mathcal H_h\) have side length \(b_{2h}\) and are placed from left to right.

For a fixed local column \(j\), a contribution \(\ctr_j(Q)\) is called a \emph{vertical contribution} if \(Q\in\mathcal V_h\) for some \(h\), and a \emph{horizontal contribution} if \(Q\in\mathcal H_h\) for some \(h\).
\end{definition}

By the recursive construction of the Euclidean Square Partition, the
vertical and horizontal layers in Definition~\ref{def:layers-contributions}
exhaust all Euclidean squares of \(\mathcal P(s,b_1)\). Fig.~\ref{fig:decomp} illustrates this decomposition of
\(\mathcal P(s,b_1)\), together with the notations \(C_r\) and \(Y_h\).


\begin{figure*}[!t]
\centering
\begin{adjustbox}{max width=\linewidth}
\begin{tikzpicture}[thick]
\setcounter{subfigure}{0}

\begin{scope}[
    local bounding box=LeftPic,
    braceL/.style={
        decorate,
        decoration={brace, amplitude=6pt, mirror, raise=4pt},
        thick
    },
    labeltextL/.style={
        midway,
        yshift=-18pt,
        font=\large
    },
    dashconnL/.style={
        dashed,
        thick,
        black!70
    },
    cguideL/.style={
        color=red,
        line width=0.95pt,
        dash pattern=on 3pt off 3pt
    },
    yguideL/.style={
        color=green,
        line width=0.95pt,
        dash pattern=on 3pt off 3pt
    },
    clabelL/.style={
        font=\large,
        text=red
    },
    ylabelL/.style={
        font=\large,
        text=green
    }
]

\def\LsizeL{2.0}
\def\LsizeM{1.2}
\def\LsizeS{0.8}
\def\LsameGap{0.35}
\def\LdiffGap{0.55}
\def\LdashLen{0.40}
\pgfmathsetmacro{\LhalfDashLen}{\LdashLen/2}

\coordinate (Pos1) at (0,0);
\draw (Pos1) rectangle ++(\LsizeL,\LsizeL);
\draw [braceL] (Pos1) -- ++(\LsizeL,0) node [labeltextL] {$b_2$};

\coordinate (Pos2) at ($(Pos1) + (\LsizeL+\LsameGap,0)$);

\coordinate (A1t) at ($(Pos1) + (\LsizeL,\LsizeL)$);
\coordinate (B1t) at ($(Pos2) + (0,\LsizeL)$);
\coordinate (M1t) at ($(A1t)!0.5!(B1t)$);
\draw[dashconnL] ($(M1t)+(-\LhalfDashLen,0)$) -- ($(M1t)+(\LhalfDashLen,0)$);

\coordinate (A1b) at ($(Pos1) + (\LsizeL,0)$);
\coordinate (B1b) at ($(Pos2) + (0,0)$);
\coordinate (M1b) at ($(A1b)!0.5!(B1b)$);
\draw[dashconnL] ($(M1b)+(-\LhalfDashLen,0)$) -- ($(M1b)+(\LhalfDashLen,0)$);

\draw (Pos2) rectangle ++(\LsizeL,\LsizeL);
\draw [braceL] (Pos2) -- ++(\LsizeL,0) node [labeltextL] {$b_2$};

\coordinate (Pos2b) at ($(Pos2) + (\LsizeL,0)$);
\draw (Pos2b) rectangle ++(\LsizeM,\LsizeM);
\draw [braceL] (Pos2b) -- ++(\LsizeM,0) node [labeltextL] {$b_4$};

\coordinate (Pos3) at ($(Pos2b) + (\LsizeM+\LsameGap,0)$);

\coordinate (A2t) at ($(Pos2b) + (\LsizeM,\LsizeM)$);
\coordinate (B2t) at ($(Pos3)  + (0,\LsizeM)$);
\coordinate (M2t) at ($(A2t)!0.5!(B2t)$);
\draw[dashconnL] ($(M2t)+(-\LhalfDashLen,0)$) -- ($(M2t)+(\LhalfDashLen,0)$);

\coordinate (A2b) at ($(Pos2b) + (\LsizeM,0)$);
\coordinate (B2b) at ($(Pos3)  + (0,0)$);
\coordinate (M2b) at ($(A2b)!0.5!(B2b)$);
\draw[dashconnL] ($(M2b)+(-\LhalfDashLen,0)$) -- ($(M2b)+(\LhalfDashLen,0)$);

\draw (Pos3) rectangle ++(\LsizeM,\LsizeM);
\draw [braceL] (Pos3) -- ++(\LsizeM,0) node [labeltextL] {$b_4$};

\coordinate (Pos4) at ($(Pos3) + (\LsizeM+\LdiffGap,0)$);

\node[inner sep=0pt, outer sep=0pt, font=\normalsize]
  at ($(Pos3)+(\LsizeM+\LdiffGap/2,\LsizeS/2)$) {$\cdots$};

\draw (Pos4) rectangle ++(\LsizeS,\LsizeS);
\draw [braceL] (Pos4) -- ++(\LsizeS,0) node [labeltextL] {$b_{2\nu}$};

\coordinate (Pos4b) at ($(Pos4) + (\LsizeS+\LsameGap,0)$);

\coordinate (A3t) at ($(Pos4)  + (\LsizeS,\LsizeS)$);
\coordinate (B3t) at ($(Pos4b) + (0,\LsizeS)$);
\coordinate (M3t) at ($(A3t)!0.5!(B3t)$);
\draw[dashconnL] ($(M3t)+(-\LhalfDashLen,0)$) -- ($(M3t)+(\LhalfDashLen,0)$);

\coordinate (A3b) at ($(Pos4)  + (\LsizeS,0)$);
\coordinate (B3b) at ($(Pos4b) + (0,0)$);
\coordinate (M3b) at ($(A3b)!0.5!(B3b)$);
\draw[dashconnL] ($(M3b)+(-\LhalfDashLen,0)$) -- ($(M3b)+(\LhalfDashLen,0)$);

\draw (Pos4b) rectangle ++(\LsizeS,\LsizeS);
\draw [braceL] (Pos4b) -- ++(\LsizeS,0) node [labeltextL] {$b_{2\nu}$};

\coordinate (COneBase) at ($(Pos2)+(\LsizeL,\LsizeL)$);
\coordinate (CTwoBase) at ($(Pos3)+(\LsizeM,\LsizeM)$);
\coordinate (CNuBase)  at ($(Pos4b)+(\LsizeS,\LsizeS)$);

\coordinate (COneTop) at ($(COneBase)+(0,1.65)$);
\coordinate (CTwoTop) at ($(CTwoBase)+(0,3.00)$);
\coordinate (CNuTop)  at ($(CNuBase) +(0,4.25)$);

\draw[cguideL] (COneBase) -- (COneTop)
    node[clabelL, above] {$C_1$};

\draw[cguideL] (CTwoBase) -- (CTwoTop)
    node[clabelL, above] {$C_2$};

\draw[cguideL] (CNuBase) -- (CNuTop)
    node[clabelL, above] {$C_\nu$};

\draw[yguideL] (COneBase) -- ($(CNuBase |- COneBase)$)
    node[ylabelL, right=4pt] {$Y_1$};

\draw[yguideL] (CTwoBase) -- ($(CNuBase |- CTwoBase)$)
    node[ylabelL, right=4pt] {$Y_2$};

\end{scope}

\node[font=\footnotesize, align=center] at ($(LeftPic.south)+(0,-0.8)$) {%
  \refstepcounter{subfigure}\phantomsection\label{fig:xxxx}%
  (\thesubfigure)\quad Horizontal layers%
};

\def\RsizeL{2.0}
\def\RsizeM{1.2}
\def\RsizeS{0.8}
\def\RvGap{0.4}
\def\RjumpGap{0.55}

\pgfmathsetmacro{\RHeight}{2*\RsizeL + 2*\RsizeM + 2*\RsizeS + 3*\RvGap + \RjumpGap}
\def\LRGap{3.0}

\def\RBottomLift{-0.15}

\begin{scope}[
    shift={($(LeftPic.south east)+(\LRGap,\RHeight+\RBottomLift)$)},
    local bounding box=RightPic,
    braceR/.style={
        decorate,
        decoration={brace, amplitude=6pt, raise=4pt},
        thick
    },
    labeltextR/.style={
        midway,
        xshift=12pt,
        anchor=west,
        font=\large
    },
    dashconnR/.style={dashed, thick}
]

\def\RdashLen{0.18}
\pgfmathsetmacro{\RhalfDashLen}{\RdashLen/2}

\coordinate (TopB2_TR) at (0,0);
\draw (TopB2_TR) rectangle (-\RsizeL, -\RsizeL);
\draw [braceR] (TopB2_TR) -- (0, -\RsizeL) node [labeltextR] {$b_1$};

\coordinate (BotB2_TR) at (0, -\RsizeL - \RvGap);
\draw (BotB2_TR) rectangle (-\RsizeL, -\RsizeL - \RvGap - \RsizeL);
\draw [braceR] (BotB2_TR) -- (0, -\RsizeL - \RvGap - \RsizeL) node [labeltextR] {$b_1$};

\coordinate (B2a_BL) at (-\RsizeL, -\RsizeL);
\coordinate (B2b_TL) at (-\RsizeL, -\RsizeL - \RvGap);
\coordinate (M1L) at ($(B2a_BL)!0.5!(B2b_TL)$);
\draw [dashconnR] ($(M1L)+(0,-\RhalfDashLen)$) -- ($(M1L)+(0,\RhalfDashLen)$);

\coordinate (B2a_BR) at (0, -\RsizeL);
\coordinate (B2b_TRR) at (0, -\RsizeL - \RvGap);
\coordinate (M1R) at ($(B2a_BR)!0.5!(B2b_TRR)$);
\draw [dashconnR] ($(M1R)+(0,-\RhalfDashLen)$) -- ($(M1R)+(0,\RhalfDashLen)$);

\coordinate (B4a_TR) at (0, -\RsizeL - \RvGap - \RsizeL);
\draw (B4a_TR) rectangle (-\RsizeM, -\RsizeL - \RvGap - \RsizeL - \RsizeM);
\draw [braceR] (B4a_TR) -- ++(0, -\RsizeM) node [labeltextR] {$b_3$};

\coordinate (B4b_TR) at (0, -\RsizeL - \RvGap - \RsizeL - \RsizeM - \RvGap);
\draw (B4b_TR) rectangle (-\RsizeM, -\RsizeL - \RvGap - \RsizeL - \RsizeM - \RvGap - \RsizeM);
\draw [braceR] (B4b_TR) -- ++(0, -\RsizeM) node [labeltextR] {$b_3$};

\coordinate (B4a_BL) at (-\RsizeM, -\RsizeL - \RvGap - \RsizeL - \RsizeM);
\coordinate (B4b_TL) at (-\RsizeM, -\RsizeL - \RvGap - \RsizeL - \RsizeM - \RvGap);
\coordinate (M2L) at ($(B4a_BL)!0.5!(B4b_TL)$);
\draw [dashconnR] ($(M2L)+(0,-\RhalfDashLen)$) -- ($(M2L)+(0,\RhalfDashLen)$);

\coordinate (B4a_BR) at (0, -\RsizeL - \RvGap - \RsizeL - \RsizeM);
\coordinate (B4b_TRR) at (0, -\RsizeL - \RvGap - \RsizeL - \RsizeM - \RvGap);
\coordinate (M2R) at ($(B4a_BR)!0.5!(B4b_TRR)$);
\draw [dashconnR] ($(M2R)+(0,-\RhalfDashLen)$) -- ($(M2R)+(0,\RhalfDashLen)$);

\coordinate (B4_End) at (0, -\RsizeL - \RvGap - \RsizeL - \RsizeM - \RvGap - \RsizeM);

\coordinate (S1_TR) at ($(B4_End) + (0, -\RjumpGap)$);

\def\VdotRaise{0.13}
\coordinate (VdotsPos) at ($ (B4_End)!0.5!(S1_TR) + (-\RsizeS/2,\VdotRaise) $);
\node[inner sep=0pt, outer sep=0pt, font=\normalsize] at (VdotsPos) {$\vdots$};

\draw (S1_TR) rectangle ++(-\RsizeS, -\RsizeS);
\draw [braceR] (S1_TR) -- ++(0, -\RsizeS) node [labeltextR] {$b_{2\nu'-1}$};

\coordinate (S2_TR) at ($(S1_TR) + (0, -\RsizeS - \RvGap)$);

\coordinate (S1_BL) at ($(S1_TR) + (-\RsizeS, -\RsizeS)$);
\coordinate (S2_TL) at ($(S2_TR) + (-\RsizeS, 0)$);
\coordinate (M3L) at ($(S1_BL)!0.5!(S2_TL)$);
\draw [dashconnR] ($(M3L)+(0,-\RhalfDashLen)$) -- ($(M3L)+(0,\RhalfDashLen)$);

\coordinate (S1_BR) at ($(S1_TR) + (0, -\RsizeS)$);
\coordinate (S2_TRR) at (S2_TR);
\coordinate (M3R) at ($(S1_BR)!0.5!(S2_TRR)$);
\draw [dashconnR] ($(M3R)+(0,-\RhalfDashLen)$) -- ($(M3R)+(0,\RhalfDashLen)$);

\draw (S2_TR) rectangle ++(-\RsizeS, -\RsizeS);
\draw [braceR] (S2_TR) -- ++(0, -\RsizeS) node [labeltextR] {$b_{2\nu'-1}$};

\end{scope}

\node[font=\footnotesize, align=center] at ($(RightPic.south)+(0,-0.8)$) {%
  \refstepcounter{subfigure}\phantomsection\label{fig:yyyyy}%
  (\thesubfigure)\quad Vertical layers%
};

\end{tikzpicture}
\end{adjustbox}
\caption{Decomposition of
\(\mathcal P(s,b_1)\).
Subfigure~(\ref{fig:xxxx}) groups the horizontal layers, and
subfigure~(\ref{fig:yyyyy}) groups the vertical layers. The red
dashed lines mark the column boundaries $C_r$, and the green
dashed lines mark the row boundaries~$Y_h$.}
\label{fig:decomp}
\end{figure*}

With the above notation and definitions, we can now determine $\mathsf{Exp}(j)$ for
each column~$j$ of $R_1$.

\begin{lemma}
\label{lem:column-path}
Fix $j\in[0,b_1-1]$.
\begin{enumerate}
\item[(1)] If \(C_{r-1}\le j<C_r\) for some \(1\le r\le \nu\), then column \(j\) crosses all squares in the first \(r\) vertical layers and exactly one square in the \(r\)-th horizontal layer. More precisely,
\[
\Exp(j)= \left\{ Y_{h-1}+wb_{2h-1}+j^{(2h-2)}:\ 1\le h\le r,\ 0\le w\le a_{2h-1}-1 \right\}\cup \{Y_r+j^{(2r)}\}.
\]

\item[(2)] If \(m\) is odd and \(C_\nu\le j\le b_1-1\), then column \(j\) crosses all squares in the \(\nu'\) vertical layers and no horizontal-layer square. In this case,
\[
\Exp(j)= \left\{ Y_{h-1}+wb_{2h-1}+j^{(2h-2)}:\ 1\le h\le \nu',\ 0\le w\le a_{2h-1}-1 \right\}.
\]
\end{enumerate}
\end{lemma}

\begin{proof}

Let
\[
Q=[\rho,\rho+\lambda-1]\times[\kappa,\kappa+\lambda-1]
\]
be a square of \(\mathcal{P}(s,b_1)\) crossed by column \(j\). By the definition of \(\ctr_j(Q)\), the exponent {of the entries in} \(Q\) {lying in column} \(j\) {of \(R_1\)} is
\begin{equation}\label{eq:contribution-general}
\ctr_j(Q)=\rho+j-\kappa.
\end{equation}

We record a simple consequence of the recursive definition of
\(j^{(2r)}\) in Eq.~\eqref{eq:j-iter}. For \(0\le r\le \nu\), it holds that
\begin{equation}\label{eq:j-coordinate}
j^{(2r)}=j-C_r,
\qquad\text{whenever } C_r\le j\le b_1-1.
\end{equation}
Indeed, {we prove \eqref{eq:j-coordinate} by induction on \(r\).} {The} case \(r=0\) is immediate. If the identity holds for
\(r-1\) and \(C_r\le j\le b_1-1\), then
\[
j^{(2r-2)}=j-C_{r-1}
=a_{2r}b_{2r}+(j-C_r).
\]
Note that
\(b_1-C_{r-1}=b_{2r-1}\) and \(C_r-C_{r-1}=a_{2r}b_{2r}\), so
\[
0 \leq j-C_r<b_1-C_r=b_{2r-1}-a_{2r}b_{2r}<b_{2r}.
\]
Hence
\[
j^{(2r)}
=j^{(2r-2)}\pmod {b_{2r}}
=j-C_r.
\]

First, consider a vertical-layer square \(V_{h,w}\) crossed by column \(j\), i.e., \(C_{h-1} \leq j \leq b_1-1\). By Eq.~\eqref{eq:j-coordinate} with \(r=h-1\), we have
\[
j=C_{h-1}+j^{(2h-2)}.
\]
Substituting
\[
\rho=Y_{h-1}+wb_{2h-1} \text{ and }
\kappa=C_{h-1}
\]
into Eq.~\eqref{eq:contribution-general}, we obtain the vertical contribution
\begin{equation}\label{eq:vertical-ctr}
\ctr_j(V_{h,w})
=
Y_{h-1}+wb_{2h-1}+j^{(2h-2)}.
\end{equation}

Next, consider a horizontal-layer square \(H_{h,w}\) crossed by column \(j\), i.e., $$C_{h-1}+wb_{2h}\le j \le C_{h-1}+(w+1)b_{2h}-1.$$ 
By Eq.~\eqref{eq:j-coordinate} with \(r=h-1\), we have
\[
j^{(2h-2)}=j-C_{h-1}.
\]
Since \(j-C_{h-1}\in[wb_{2h},(w+1)b_{2h}-1]\), the recursive definition gives
\[
j=C_{h-1}+wb_{2h}+j^{(2h)}.
\]
Substituting
\[
\rho=Y_h,
 \text{ and }
\kappa=C_{h-1}+wb_{2h}
\]
into Eq.~\eqref{eq:contribution-general}, we obtain the horizontal contribution
\begin{equation}\label{eq:horizontal-ctr}
\ctr_j(H_{h,w})=Y_h+j^{(2h)}.
\end{equation}

\noindent\emph{(1) $C_{r-1}\le j<C_r$ for some $1\le r\le\nu$.}

Then \(j\) lies in the column range of every vertical layer \(\mathcal V_h\) with \(1\le h\le r\), and in the column range of no later vertical layer. Moreover, it lies in exactly one square of the \(r\)-th horizontal layer and in no horizontal layer with index different from \(r\). Substituting {Eqs}.~\eqref{eq:vertical-ctr} and~\eqref{eq:horizontal-ctr} {into} Eq.~\eqref{eq:Exp-ctr} {gives} the first formula.

\smallskip
\noindent\emph{(2) $m$ odd and $C_\nu\le j\le b_1-1$.}

Then \(\nu'=\nu+1\), and column \(j\) lies in the column range of every vertical layer \(\mathcal V_h\), \(1\le h\le \nu'\), while it lies in no horizontal layer. Substituting Eq.~\eqref{eq:vertical-ctr} {into} Eq.~\eqref{eq:Exp-ctr} gives the second formula.
\end{proof}

Finally, substituting the results of Lemma~\ref{lem:column-path}(1) and (2) into Eq.~\eqref{eq:Rj-exp}, respectively, we obtain

\begin{enumerate}
\item[(1)] If $1\le r\le \nu$ and $C_{r-1}\le j<C_r$, then
\[
\R_j={}\alpha^{a_0 s+j}\!\Biggl( \sum_{h=0}^{r-1}\sum_{w=0}^{a_{2h+1}-1} \beta^{Y_{h} +w\,b_{2h+1}+j^{(2h)}}+ \beta^{Y_r+j^{(2r)}} \Biggr).
\]

\item[(2)] If $m$ is odd and $C_{\nu}\le j\le b_1-1$, then
\[
\R_j
=
\alpha^{a_0s+j}
\sum_{h=0}^{\nu'-1}\sum_{w=0}^{a_{2h+1}-1}
\beta^{Y_h
  +w\,b_{2h+1}+j^{(2h)}}, 
\]
\end{enumerate}
which coincide with the formulas stated in Theorem~\ref{thm:main}.

\section{Proofs of Lemmas~\ref{lem:pred-preceding-fixed} and~\ref{lem:subdiag}}\label{app:proof-lem-main}

This appendix proves Lemmas~\ref{lem:pred-preceding-fixed} and~\ref{lem:subdiag}, which are used in the induction argument of Section~\ref{subsec:proof-main}. We retain the notation of Section~\ref{sec:main} and Appendix~\ref{app:explicit-basis}. In particular, \(R_1\) denotes the local subarray formed by the last \(b_1\) columns of \(R\), with local index set \(\I(s,b_1)\), and
\[
t_{i,j}:=\tau(i,a_0s+j),
\qquad (i,j)\in\I(s,b_1).
\]
The common task in both lemmas is to show that the position produced by the locating map precedes the current position \((i,j)\) in the column-major order
\[
(i',j')\prec(i,j)
\quad\Longleftrightarrow\quad
j'<j\quad\text{or}\quad j'=j\text{ and }i'<i.
\]
Lemma~\ref{lem:pred-preceding-fixed} treats the case \(t<t_{i,j}\), whereas Lemma~\ref{lem:subdiag} treats the case \(t>t_{i,j}\) with \(i+j\ge t\).

We shall use the column-wise contribution notation introduced in Definition~\ref{def:layers-contributions}. Thus, when \(t\in\Exp(j)\), we say that \(t\) is {contributed} by a square \(Q\) {in column \(j\) of \(R_1\)} if
\[
t=\ctr_j(Q).
\]
Equivalently, \(t\) is the \(\beta\)-exponent of the entries lying in the intersection of \(Q\) with the local column \(j\).

The following two facts are useful for our proofs.

\begin{fact}\label{fact:horizontal-max}
Fix $j\in[0,b_1-1]$. If $\mathsf{Exp}(j)$ contains a horizontal contribution, then that contribution is the unique maximal element of $\mathsf{Exp}(j)$. 
Consequently, for every $(i,j)\in\I(s,b_1)$ and every
\[
t\in\mathsf{Exp}(j)\setminus\{t_{i,j}\},
\]
the inequality $t<t_{i,j}$ implies that $t$ is a vertical contribution.
\end{fact}

\begin{proof}
By Lemma~\ref{lem:column-path}, the horizontal contribution comes
from a square lying strictly below
every vertical-layer square crossed by {the} local column~$j$.
By parts~(2) and~(4) of Proposition~\ref{prop:beta-structure},
each square contributes a single $\beta$-exponent {in} column~{$j$ of $R_1$},
and lower squares contribute strictly larger exponents. 
Hence, in local column~$j$, the horizontal contribution, if it exists, is strictly greater than any vertical contribution and is therefore the unique maximum of $\mathsf{Exp}(j)$.
\end{proof}

\begin{fact}\label{fact:layer-boundary}
Fix $r\in[1,\nu]$ and $w\in[0,a_{2r}]$, and let
\[
  D := C_{r-1}+w\,b_{2r}.
\]
Every Euclidean square of $\mathcal{P}(s,b_1)$ whose row range meets $[Y_r,s-1]$
lies entirely in columns $[0,D-1]$ or entirely in columns $[D,b_1-1]$.
\end{fact}

\begin{proof}
By Definition~\ref{def:layers-contributions}, the vertical layers $1,\ldots,r$ of
$\mathcal{P}(s,b_1)$ have row ranges contained in $[0,Y_r-1]$,
so they do not meet $[Y_r,s-1]$.
The remaining squares that meet $[Y_r,s-1]$ fall into three groups.
\begin{itemize}
\item \emph{Horizontal layers $1,\ldots,r-1$.} For $1\leq h<r$, the $h$-th horizontal layer  has column range
  $[C_{h-1},C_h-1]\subseteq[0,C_{r-1}-1]\subseteq[0,D-1]$
  (since $D\ge C_{r-1}$).
\item \emph{The $r$-th horizontal layer.}
  It partitions $[Y_r,s-1]\times[C_{r-1},C_r-1]$ into $a_{2r}$ squares
  of width~$b_{2r}$; each occupies columns
  $[C_{r-1}+w'b_{2r},\,C_{r-1}+(w'+1)b_{2r}-1]$ for some
  $w'\in[0,a_{2r}-1]$, and lies entirely on one side of $D$ by the definition of $D$.
\item \emph{All squares after the $r$-th horizontal layer.}
  Their column ranges are contained in $[C_r,b_1-1]\subseteq[D,b_1-1]$
  (since $D\le C_r$).
\end{itemize}
\end{proof}

\begin{proof}[\textbf{Proof of Lemma~\ref{lem:pred-preceding-fixed}}]
By Fact~\ref{fact:horizontal-max}, $t<t_{i,j}$ implies that $t$ is a
vertical contribution. Hence there exist \(h\in[1,\nu']\) and \(w\in[0,a_{2h-1}-1]\) such that {\(V_{h,w}\) contributes} \(t\) {in column \(j\) of \(R_1\). Let}
\[
Q^*=V_{h,w}=[Y_{h-1}+wb_{2h-1},\,Y_{h-1}+(w+1)b_{2h-1}-1]\times[C_{h-1},\,b_1-1].
\]
Thus
\begin{equation}\label{eq:lemma8-tj}
t=Y_{h-1}+wb_{2h-1}+j^{(2h-2)},\qquad j=C_{h-1}+j^{(2h-2)}.
\end{equation}

Fix \(\xi\in[0,b_1-1]\), and let \(Q'\) be the square of \(\mathcal{P}(s,b_1)\) containing \((t,\xi)\).

\smallskip
\emph{Case~1: $\xi<C_{h-1}$.}
This case is void when $h=1$, so assume $h\ge 2$. Since \(t\ge Y_{h-1}\), the row range of \(Q'\) meets \([Y_{h-1},s-1]\). Applying Fact~\ref{fact:layer-boundary} with \(r=h-1\) and \(D=C_{h-1}\), the square \(Q'\) lies entirely on one side of column \(C_{h-1}\). 
The assumption $\xi<C_{h-1}$ forces it onto the left side, giving
column range $\subseteq[0,C_{h-1}-1]$.
By~Eq.~\eqref{eq:locsquare}, $(i_{t,\xi},j_{t,\xi})\in Q'$, hence $j_{t,\xi}\le C_{h-1}-1<C_{h-1}\le j$.
The column priority in $\prec$ yields
$\operatorname{Loc}(t,\xi)\prec(i,j)$.

\smallskip
\emph{Case~2: $\xi\ge C_{h-1}$.} 
In this case, by Eq.~\eqref{eq:lemma8-tj}, $t$ lies in the row range of $Q^*$, since
$$j^{(2h-2)}=j-C_{h-1}\le b_1-1-C_{h-1}=b_{2h-1}-1.$$
Moreover, $\xi$
lies in $[C_{h-1},b_1-1]$ by assumption. Hence \((t,\xi)\in Q^*\). The explicit locating formula from Lemma~\ref{lem:layer-segment} gives
\[
j_{t,\xi}
=
C_{h-1}+\bigl(t-(Y_{h-1}+wb_{2h-1})\bigr)
=
C_{h-1}+j^{(2h-2)}
=
j,
\]
and
\[
i_{t,\xi}
=
Y_{h-1}+wb_{2h-1}+(\xi-C_{h-1}).
\]
Since \(\xi\le b_1-1=C_{h-1}+b_{2h-1}-1\), we have
\begin{equation}\label{eq:lemma8-itxi-upper}
i_{t,\xi}
\le
Y_{h-1}+(w+1)b_{2h-1}-1.
\end{equation}
On the other hand, \(t<t_{i,j}\), together with Proposition~\ref{prop:beta-structure}(4), implies that the square containing \((i,j)\) lies strictly below \(Q^*\). Hence
\begin{equation}\label{eq:lemma8-i-lower}
i\ge Y_{h-1}+(w+1)b_{2h-1}.
\end{equation}
Combining Eq.~\eqref{eq:lemma8-itxi-upper} and Eq.~\eqref{eq:lemma8-i-lower}, we obtain \(i_{t,\xi}<i\). Since \(j_{t,\xi}=j\), the row tie-breaker in the column-major order gives
\[
\Loc(t,\xi)=(i_{t,\xi},j_{t,\xi})\prec(i,j).
\]
\end{proof}

\begin{proof}[\textbf{Proof of Lemma~\ref{lem:subdiag}}]
From $\delta=i+j-t$ and $i+j\ge t$, we have $\delta\ge 0$.
It remains to show $\delta\le b_1-1$ and $\operatorname{Loc}(t,\delta)\prec(i,j)$.

Since $t>t_{i,j}$, Proposition~\ref{prop:beta-structure}(4) implies that the square {contributing} \(t\) in column \(j\) {of \(R_1\)} lies strictly below the square containing \((i,j)\). We distinguish whether \(t\) is a vertical or horizontal contribution.

\smallskip
\emph{Vertical contribution.}
Suppose that \(t\) is {contributed} by a vertical-layer square \(V_{h,w}\) {in column \(j\) of \(R_1\)}. By Lemma~\ref{lem:column-path}, 
\[
t=Y_{h-1}+wb_{2h-1}+j^{(2h-2)},\qquad j=C_{h-1}+j^{(2h-2)}.
\]
Since the square containing \((i,j)\) lies strictly above \(V_{h,w}\), we have
\[
i\le Y_{h-1}+wb_{2h-1}-1.
\]
Therefore
\[
\delta= i+j-t= i+C_{h-1}-(Y_{h-1}+wb_{2h-1})\le C_{h-1}-1 \le b_1-1.
\]
Together with \(\delta\ge0\), this also shows that the present case can occur only when \(h\ge2\). Moreover, \(\delta<C_{h-1}\le j\). Since \(t\ge Y_{h-1}\), the square containing \((t,\delta)\) has row range meeting \([Y_{h-1},s-1]\). By Fact~\ref{fact:layer-boundary}, applied with \(r=h-1\) and \(D=C_{h-1}\), this square lies entirely to the left of column \(C_{h-1}\). Therefore the column coordinate of \(\Loc(t,\delta)\) is at most \(C_{h-1}-1<j\), and hence
\[
\Loc(t,\delta)\prec(i,j).
\]

\smallskip
\emph{Horizontal contribution.}
Suppose that \(t\) is {contributed} by a horizontal-layer square \(H_{r,w}\) {in column \(j\) of \(R_1\)}. Put
\[
\kappa:=C_{r-1}+wb_{2r}.
\]
By Lemma~\ref{lem:column-path}, 
\[
t=Y_r+j^{(2r)},\qquad j=\kappa+j^{(2r)}.
\]
Since the square containing \((i,j)\) lies strictly above \(H_{r,w}\), we have
\[
i\le Y_r-1.
\]
Consequently,
\[
\delta
=
i+j-t
=
i+\kappa-Y_r
\le \kappa-1
\le b_1-1.
\]
Again \(\delta\ge0\), and hence \(\delta<\kappa\le j\). Since \(t\ge Y_r\), the square containing \((t,\delta)\) has row range meeting \([Y_r,s-1]\). Applying Fact~\ref{fact:layer-boundary} with \(D=\kappa\), this square lies entirely to the left of column \(\kappa\). Therefore the column coordinate of \(\Loc(t,\delta)\) is at most \(\kappa-1<j\), and hence
\[
\Loc(t,\delta)\prec(i,j).
\]
This completes the proof.
\end{proof}

The following example provides a detailed walkthrough of the inductive recovery argument for a concrete instance, exercising every component of the proof machinery developed in Lemmas~\ref{lem:Rformulas}--{\ref{lem:column-path}}.

\begin{example}\label{ex:induction-illustration}

Let $p=70$, $s=29$. We show how $R(14,68)$ can be recovered from $\overline{R}$.

The Euclidean algorithm gives
\[
70=2\cdot 29+12,\quad 29=2\cdot 12+5,\qquad 12=2\cdot 5+2,\quad 5=2\cdot 2+1,\quad 2=2\cdot 1,
\]
so
\[
\begin{gathered}
a_0=2,\quad a_1=a_2=a_3=a_4=2,\\
b_1=12,\quad b_2=5,\quad b_3=2,\quad b_4=1,
\end{gathered}
\]
and $a_0s=58$, $\nu=\lfloor m/2\rfloor=2$. By Lemma~\ref{lem:R1-partition}, the local index set of $R_1$ is
\[
\I(29,12)=[0,28]\times[0,11],
\]
and its partition is $\mathcal P(29,12)$. Thus the first vertical layer consists of two $12\times 12$ squares, the first horizontal layer consists of two $5\times5$ squares, the second vertical layer consists of two $2\times2$ squares, and the second horizontal layer consists of two $1\times1$ squares.

Applying Lemma~\ref{lem:Rformulas}(2) to these squares gives the following local $\beta$-exponent profile:
\[
\tau(i,58+j)\!=\!
\begin{cases}
j, & 0\le i\le 11,\\
12+j, & 12\le i\le 23,\\
24+(j\bmod 5), & 24\le i\le 28,\ 0\le j\le 9,\\
24+(j-10), & 24\le i\le 25,\ j\in\{10,11\},\\
26+(j-10), & 26\le i\le 27,\ j\in\{10,11\},\\
28, & i=28,\ j\in\{10,11\}.
\end{cases}
\]

Next we calculate $\mathsf{Exp}(10)$.
Consider the local column $j=10$, i.e., the global column $58+10=68$. Here
\[
C_1=a_2b_2=10,\qquad
C_2=a_2b_2+a_4b_4=12,
\]
and
\[
j^{(0)}=10,\qquad j^{(2)}=10\pmod 5=0,\qquad j^{(4)}=0\pmod 1=0.
\]
Since $C_1\le 10<C_2$, Lemma~\ref{lem:column-path}(1) with $r=2$ applies. The vertical contributions are
\[
0+0\cdot 12+10=10,\qquad
0+1\cdot 12+10=22,
\]
from the two $12\times12$ squares, and
\[
24+0\cdot 2+0=24,\qquad
24+1\cdot 2+0=26,
\]
from the two $2\times2$ squares. The second horizontal layer contributes
\[
28+j^{(4)}=28.
\]
Therefore
\[
\mathsf{Exp}(10)=\{10,22,24,26,28\}.
\]

\emph{The target local point.}
Take $(i,j)=(14,10)$. By the exponent profile above,
\[
t_{14,10}=\tau(14,68)=22,
\]
and by Lemma~\ref{lem:Rformulas}(1),
\[
R(14,68)=\alpha^{58+14+10}\beta^{22}=\alpha^{82}\beta^{22}.
\]
Thus {Eq.~\eqref{eq:recover}} becomes
\[
R(14,68)
=
\overline R(14,68)
-
\alpha^{82}\bigl(\beta^{10}+\beta^{24}+\beta^{26}+\beta^{28}\bigr).
\]
Moreover, $\overline R(14,68)\in K$ by Eq.~\eqref{eq:AggSum}. It remains to show that
\[
\alpha^{82}\beta^t\in K
\qquad
\text{for }t\in\{10,24,26,28\}.
\]
The four interference exponents fall into the three cases exactly as follows:
\[
\begin{array}{c|c|c}
 t & \text{condition} & \text{case} \\ \hline
 26,28 & 14+10<t & \text{(C1)}\\
10 & t<t_{14,10}=22 & \text{(C2)} \\
24 & t>t_{14,10}=22\text{ and }14+10\ge 24 & \text{(C3)} 
\end{array}
\]

\emph{Case (C1): $t=26,28$.}
For $t\in\{26,28\}$, set
\[
q_t:=58+14+10-t=82-t.
\]
Then $q_{26}=56$ and $q_{28}=54$, both lying in $[0,a_0s-1]=[0,57]$. Lemma~\ref{lem:initial}, applied with $\theta=t$, gives
\[
\alpha^{q_t}(\alpha\beta)^t
=
\alpha^{82}\beta^t
\in K.
\]
Hence $\alpha^{82}\beta^{26},\alpha^{82}\beta^{28}\in K$.

\emph{Case (C2): $t=10$.}
Here $10<t_{14,10}$, so Lemma~\ref{lem:pred-preceding-fixed} applies. We spell out the locating map to make the use of Lemma~\ref{lem:layer-segment} explicit. For every $\xi\in[0,11]$, the point $(10,\xi)$ lies in the top-left square $[0,11]\times[0,11]$ of $\mathcal P(29,12)$. Thus~Eq.~\eqref{eq:locsquare} gives
\[
\operatorname{Loc}(10,\xi)=(\xi,10).
\]
Since $\xi\le 11<14$, we have $(\xi,10)\prec(14,10)$ for every $\xi\in[0,11]$. By the induction hypothesis and Lemma~\ref{lem:layer-segment},
\[
R(\xi,58+10) = \alpha^{58+10+\xi}\beta^{10} = \alpha^{68+\xi}\beta^{10} \in K,\qquad 0\le \xi\le 11.
\]
On the other hand, Lemma~\ref{lem:initial} gives
\[
\alpha^{10+u}\beta^{10}\in K,
\qquad 0\le u\le 57.
\]
Together these are the $70=p$ consecutive elements
\[
\alpha^{10}\beta^{10},\alpha^{11}\beta^{10},\ldots,\alpha^{79}\beta^{10}
\]
in $K$. Lemma~\ref{lem:field-gen} yields $\E\beta^{10}\subseteq K$, and therefore $\alpha^{82}\beta^{10}\in K$.

\emph{Case (C3): $t=24$.}
Here
\[
\delta:=14+10-24=0\in [0,b_1-1]=[0,11].
\]
Since $(24,0)$ lies in the square $[24,28]\times[0,4]$, Eq.~\eqref{eq:locsquare} gives
\[
\operatorname{Loc}(24,0)=(24,0).
\]
This point precedes $(14,10)$ by column priority, since its column is $0<10$. By the induction hypothesis and Lemma~\ref{lem:layer-segment},
\[
R(24,58)
=
\alpha^{58+24+0}\beta^{24}
=
\alpha^{82}\beta^{24}
\in K.
\]

We have proved that all interference monomials
\[
\alpha^{82}\beta^{10},\quad
\alpha^{82}\beta^{24},\quad
\alpha^{82}\beta^{26},\quad
\alpha^{82}\beta^{28}
\]
belong to $K$. Since $\overline R(14,68)\in K$, {Eq.~\eqref{eq:recover}} gives $R(14,68)\in K$, as required by the inductive step.
\end{example}

\end{document}